\documentclass[aps,pra,superscriptaddress,twocolumn,nofootinbib]{revtex4-2} 
\pdfoutput=1 
\usepackage{dsfont} 
\usepackage{comment} 
\usepackage[utf8]{inputenc} 
\usepackage[T1]{fontenc} 
\usepackage[british]{babel} 
\usepackage[dvipsnames,x11names]{xcolor} 
\usepackage{amsmath,amssymb,amsthm,bm,amsfonts,bbm} 
\usepackage{mathtools} 
\usepackage[colorlinks=true,citecolor=Green,linkcolor=Purple,urlcolor=Blue]{hyperref} 
\usepackage{physics}
\usepackage{amsfonts}

\newcommand{\id}{\mathds{1}}
\newcommand{\lam}{\vec{\lambda}}
\newcommand{\vai}{\vec{a}_i}

\DeclareMathOperator{\sgn}{sgn}
\newtheorem{lemma}{Lemma}

\begin{document}
\title{Compatibility of Generalized Noisy Qubit Measurements}


\author{Martin J. Renner}
    \affiliation{University of Vienna, Faculty of Physics, Vienna Center for Quantum Science and Technology (VCQ), Boltzmanngasse 5, 1090 Vienna, Austria}
    \affiliation{Institute for Quantum Optics and Quantum Information (IQOQI), Austrian Academy of Sciences, Boltzmanngasse 3, 1090 Vienna, Austria}


\date{\today}

\begin{abstract}
It is a crucial feature of quantum mechanics that not all measurements are compatible with each other. However, if measurements suffer from noise they may lose their incompatibility. Here, we consider the effect of white noise and determine the critical visibility such that all qubit measurements, i.e. all positive operator-valued measures (POVMs), become compatible, i.e. jointly measurable. In addition, we apply our methods to quantum steering and Bell nonlocality. We obtain a tight local hidden state model for two-qubit Werner states of visibility $1/2$. This determines the exact steering bound for two-qubit Werner states and also provides a local hidden variable model that improves on previously known ones. Interestingly, this proves that POVMs are not more powerful than projective measurements to demonstrate quantum steering for these states.
\end{abstract}

\maketitle

\textit{Introduction.---} Quantum mechanics provides a remarkably accurate framework for predicting the outcomes of experiments and has led to the development of numerous technological advancements. Despite its successes, it presents us with puzzling and counterintuitive phenomena that challenge our classical notions of reality. One of the key aspects that set quantum mechanics apart from classical physics is the concept of measurement incompatibility. In classical physics, measuring one property of a system need not affect the measurement of another property. In quantum mechanics, however, the situation is radically different. The uncertainty principle, formulated by Werner Heisenberg, establishes a fundamental limit to the precision with which certain pairs of properties can be simultaneously known \cite{Heisenberg1927}.

A simple and well-known example is the fact that we cannot simultaneously measure the spin of a particle in two orthogonal directions. It is known that incompatible measurements are at the core of many quantum information tasks. For example, they are necessary to violate Bell inequalities \cite{Fine1982, Fine1982b, Wolf2009} and necessary to provide an advantage in quantum communication \cite{Carmeli2020, FrenkelWeiner2015, Saha2023} or state discrimination tasks \cite{Carmeli2019, Uola2019, Skrzypczyk2019} (see also the reviews \cite{IncompatibilityReview2015, Compatibilityreview202}).

However, measurement devices always suffer from imprecision. Therefore, an apparatus measures in practice only a noisy version of the measurements. If the noise gets too large, these noisy measurements can become compatible even though they are incompatible in the noiseless limit~\cite{Busch1986}. In that case, the statistics of these noisy measurements can be obtained from the statistics of just a single measurement, and we say that these noisy measurements are jointly measurable. However, a detector that can only perform compatible measurements has limited power. Most importantly, it cannot be used for many quantum information processing tasks like demonstrating Bell-nonlocality since these require incompatible measurements. It is therefore important to ask, how much noise can be tolerated before all measurements become jointly measurable.

\begin{figure}[]
    \centering
    \includegraphics[width=0.7\columnwidth]{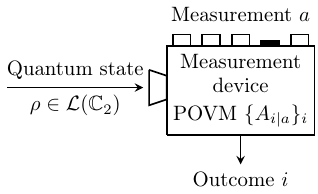}
    \caption{A measurement device can perform different measurements (labeled with $a$) that produce an outcome $i$. If the measurements are too noisy they can be simulated by a device that just performs a single measurement. In this work, we address the question of how much white noise can be tolerated before all qubit measurements become jointly measurable.} \label{fig:JM}
\end{figure}

In this work, we study the effect of white noise and show that all qubit measurements become jointly measurable at a critical visibility of $1/2$.
This result has direct implications for related fields of quantum information, in particular, Bell nonlocality \cite{Bell, Brunner_2014} and quantum steering \cite{EPR1935, Wiseman2007, He2013QKD, Bowles2014, Cavalcanti2017, SteeringReview2020, Sekatski2023}. More precisely, we use the close connection between joint measurability and quantum steering \cite{Quintino2014, Roope2014, Uola2015} to show that the two-qubit Werner state~\cite{Werner1989}
\begin{align}
    \rho^\eta_{W}=\eta\ \ketbra{\Psi^-} + (1-\eta)\ \id/4
\end{align}
cannot demonstrate EPR-steering if $\eta\leq 1/2$. Here, $\ket{\Psi^-}=\frac{1}{\sqrt{2}}(\ket{01}-\ket{10})$ denotes the two-qubit singlet state. This also implies that the same state does not violate any Bell inequality for arbitrary positive operator-valued measurements (POVM) applied on both sides whenever $\eta\leq 1/2$.

\textit{Notation and joint measurability.---} Before we introduce the problem, we introduce the necessary notation. Qubit states are described by positive semidefinite $2\times 2$ complex operators $\rho\in\mathcal{L}(\mathbb{C}_{2})$, $\rho\geq 0$ with unit trace $\tr[\rho]=1$. They can be represented as $\rho=\left(\id+\vec{x}\cdot \vec{\sigma}\right)/2$, where $\vec{x}\in\mathbb{R}^3$ is a three-dimensional real vector such that  $|\vec{x}|\leq 1$, and $\vec{\sigma}=(\sigma_x,\sigma_y,\sigma_z)$ are the standard Pauli matrices. In this notation, $\vec{x}$ is the corresponding Bloch vector of the qubit state. General qubit measurements are described by a positive operator-valued measure (POVM), which is a set of positive semidefinite operators $A_{i|a}\geq 0$ that sum to the identity, $\sum_i A_{i|a}=\mathds{1}$. Here, we use the label $"a"$ to distinguish between different measurements, while $"i"$ denotes the outcome of a given POVM (see also Fig.~\ref{fig:JM}). In quantum theory, the probability of outcome $i$ when performing the POVM with elements $A_{i|a}$ on the state $\rho$ is given by Born's rule,
\begin{align}
	p(i|a,\rho)=\tr[A_{i|a}\ \rho]\, .\label{Bornlhs}
\end{align}
Because every qubit POVM can be written as a coarse-graining of rank-1 projectors~\cite{Barrett2002}, we may restrict ourselves to POVMs proportional to rank-1 projectors. (We could also restrict ourselves to POVMs with at most four outcomes \cite{DAriano2005} but this is not necessary in what follows.) Thus, we write the measurements as $A_{i|a}=p_i \ketbra{\vec{a}_i}$, where $p_i\geq 0$ and $\ketbra{\vec{a}_i}=\big(\mathds{1} + \vec{a}_i \cdot \vec{\sigma}\big)/2$ for some normalized vector $\vec{a}_i\in\mathbb{R}^3$ ($|\vai|=1$). As a consequence of $\sum_i A_{i|a}=\mathds{1}$ we obtain $\sum_i p_i=2$ and $\sum_i p_i\ \vai = \vec{0}$.

These expressions are valid if all measurements are perfectly implemented. However, noise is usually unavoidable in experiments. In this work, we study the effect of white noise, where $\eta$ denotes the visibility. More formally, we define the noisy measurements as:
\begin{align}
    A^\eta_{i|a}=\eta\ A_{i|a} + (1-\eta) \ \tr[A_{i|a}]\ \id/2 \, .
\end{align}
With the notation introduced above the POVM elements become $A^\eta_{i|a}=p_i \big(\mathds{1} +\eta \ \vec{a}_i \cdot \vec{\sigma}\big)/2$. The goal of this work is to determine the critical value of $\eta$ such that all qubit POVMs become jointly measurable.

A set of measurements $\{A_{i|a}\}_{i,a}$ is jointly measurable if there exists a single measurement (so-called parent POVM) $\{G_\lambda\}_\lambda$ such that the statistics of all measurements in the set can be obtained by classical post-processing of the data of that single parent measurement. More precisely, if for every POVM in the set there exist conditional probabilities $p(i|a,\lambda)$ (with $0\leq p(i|a,\lambda)\leq 1$ and $\sum_i p(i|a,\lambda)=1$) such that:
\begin{align}
    A_{i|a}=\sum_\lambda p(i|a,\lambda)\ G_\lambda \, .
\end{align}
If this is satisfied, all measurements in the set can be simulated by the single parent POVM with operators $G_\lambda$: First, the parent POVM is measured on the quantum state $\rho$ in which outcome $\lambda$ occurs with probability $p(\lambda|\rho)=\tr[G_\lambda\ \rho]$. Second, given the POVM labeled by "a" that we want to simulate, the outcome $i$ is produced with probability $p(i|a,\lambda)$. In total, the probability of outcome $i$ becomes:
\begin{align}
    \sum_\lambda p(i|a,\lambda) p(\lambda|\rho)=\sum_\lambda p(i|a,\lambda) \tr[G_\lambda\ \rho]= \tr[A_{i|a}\ \rho] \, .
\end{align}
Here, we used the linearity of the trace. This perfectly simulates a given POVM with elements $\{A_{i|a}\}_i$, since this is the same expression as if the measurement was directly performed on the quantum state $\rho$ given in Eq.~\eqref{Bornlhs}.

The most prominent example are the two noisy spin-measurements $A^\eta_{\pm|x}=(\id\pm 1/\sqrt{2}\  \sigma_x)/2$ and $A^\eta_{\pm|z}=(\id\pm 1/\sqrt{2}\ \sigma_z)/2$ where $\eta=1/\sqrt{2}$. We can consider the following measurement with four outcomes $\lambda=(i,j)$ where $i,j\in\{+1,-1\}$:
\begin{align}
    G_{(i,j)}=\frac{1}{4}\left(\id +\frac{i}{\sqrt{2}} \sigma_x+\frac{j}{\sqrt{2}} \sigma_z\right)\, .
\end{align}
One can check that this is a valid POVM and that $A^\eta_{i|x}=\sum_j G_{(i,j)}$ as well as $A^\eta_{j|z}=\sum_i G_{(i,j)}$. Therefore, the statistics of both measurements $\{A^\eta_{i|x}\}_i$ and $\{A^\eta_{j|z}\}_j$ can be obtained from the statistics of just a single parent measurement. Now we consider not only two but the set of all noisy qubit POVMs $\{A^\eta_{i|a}\}_{i,a}$ and show that for $\eta\leq 1/2$ this set becomes jointly measurable.



\begin{figure*}[hbt!]
    \centering
    \includegraphics[width=1.0\textwidth]{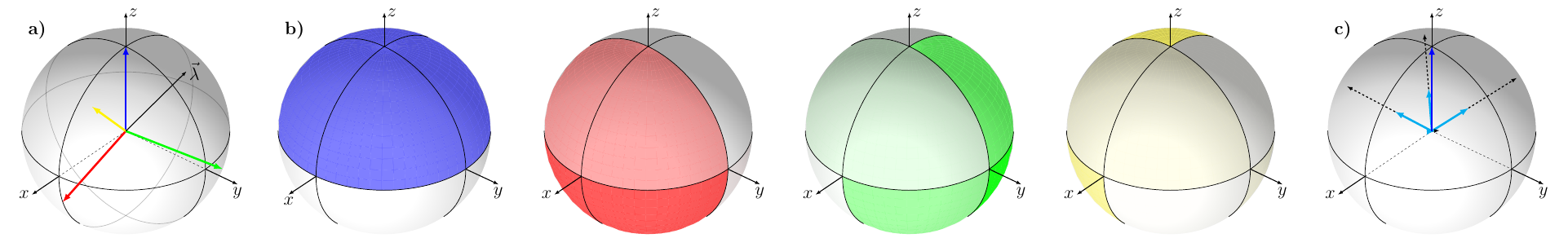}
    \caption{An illustration for a SIC-POVM \cite{SIC}: \textbf{a)} The different outcomes $i$ are represented with different colors and the colored vectors represent $\vai$ (note also $p_i=1/2$ for $i=1,2,3,4$). \textbf{b)} The opacity of the colors represents the probability to output $i$ given that $\lam$ lies in that region of the sphere, hence $p(i|a, \lam)$. This function is constant in each octant of the chosen frame, which is simply the standard coordinate frame in this case. For the $\lam$ shown in the left sphere ($s_x=-1$, $s_y=s_z=+1$), the outcome is most likely blue ($50\%$) or green ($49\%$). \textbf{c)}  Collecting all results $\lam$ from one octant behaves like the operator $G_{s_xs_ys_z}$ represented by the cyan arrows for the blue outcome. The sum of these operators simulates the desired (blue) operator $A^{1/2}_{1|a}$. (more details in Appendix~\ref{SICexample})} \label{illsicmain}
\end{figure*}

\textit{The protocol.---} First, we define two functions. The first one is the sign function, which is defined as $\sgn{(x)}:=+1$ if $x\geq0$ and $\sgn{(x)}:=-1$ if $x<0$. Similarly, the function $\Theta(x)$ is defined as $\Theta(x):=x$ if $x\geq 0$ and $\Theta(x):=0$ if $x<0$ (or, $\Theta(x):=(|x|+x)/2$).

The parent POVM $\{G_{\vec{\lambda}}\}_{\lam}$ is the measurement with elements
\begin{align}
     G_{\vec{\lambda}}=\frac{1}{4\pi}(\id + \vec{\lambda}\cdot \vec{\sigma}) \, . \label{parentPOVM}
 \end{align}
Here, $\lam \in \mathbb{R}^3$ is a normalized vector uniformly distributed on the unit radius sphere $S_2$. Physically, this corresponds to a (sharp) projective measurement with outcome $\lam$, where the measurement direction is chosen Haar-random on the Bloch sphere~\cite{footnote}.

For a given POVM with operators $A^{1/2}_{i|a}=p_i \big(\mathds{1} +  1/2\ \vec{a}_i \cdot \vec{\sigma}\big)/2$ where $\sum_i p_i=2$, $|\vai|=1$, and $\sum_i p_i\ \vai = \vec{0}$, we define the following function that associates a real-valued number to each point in $\vec{x}\in \mathbb{R}^3$:
 \begin{align}
    f_a: \mathbb{R}^3 \to \mathbb{R}:\ f_a(\vec{x}):=\sum_i p_i \ \Theta (\vec{x}\cdot \vec{a}_i) \, .
\end{align}
Now, we choose an orthonormal coordinate frame of the Bloch sphere, defined by the three pairwise orthogonal unit vectors $\vec{x}',\vec{y}', \vec{z}' \in S_2$. In addition, we define the eight vectors $\vec{v}_{s_x s_y s_z}:=s_x\vec{x}'+s_y\vec{y}'+s_z\vec{z}'$ where $s_x,s_y,s_z\in \{+1,-1\}$. This frame shall be chosen such that $f_a(\vec{v}_{s_x s_y s_z})\leq 1$ for all of these eight vectors and we show below that one can always find such a coordinate frame. Note that the vectors $\vec{v}_{s_x s_y s_z}$ are the vertices of a cube with sidelength two that is centered at the origin of the Bloch sphere.


After choosing a suitable frame, we can define the conditional probabilities:
\begin{align}
    p(i|a, \lam)=p_i \ \Theta(\vec{a}_i\cdot \vec{v}_{s_xs_ys_z})+\frac{(1-f_a(\vec{v}_{s_xs_ys_z}))\alpha_i}{\sum_i \alpha_i} \, . \label{condprob}
\end{align}
Here, $\vec{v}_{s_xs_ys_z}$ is the vector with indices $s_k=\sgn{(\lam \cdot \vec{k}')}$ for $k\in \{x,y,z\}$. Hence, the three signs $s_k$ denote the octant of $\lam$ in the rotated frame defined by $\vec{x}', \vec{y}', \vec{z}'$. (Equivalently, $\vec{v}_{s_x s_y s_z}$ is the vertex of the cube closest to $\lam$.) In addition, $\alpha_i$ is defined as:
\begin{align}
    \alpha_i:=\frac{p_i}{2}\left(1-\frac{1}{4}\sum_{s_x,s_y,s_z=\pm 1} \Theta(\vec{a}_i\cdot \vec{v}_{s_xs_ys_z})\right) \, . \label{defalphai}
\end{align}

\textit{Idea of the protocol.---} Suppose for now that it is possible to find a suitable frame in which $f_a(\vec{v}_{s_x s_y s_z})\leq 1$ for all eight vectors $\vec{v}_{s_x s_y s_z}$. Since this part is more technical, we discuss it at the end of this section. We can check first that the conditional probabilities are indeed well-defined. Namely, they are positive and sum to one. Positivity follows from the fact that $p_i\geq 0$ and $\Theta(x)\geq 0$ (for all $x\in\mathbb{R}$). In addition, $f_a(\vec{v}_{s_xs_ys_z})\leq 1$ and the proof that $\alpha_i\geq 0$ is given in the Appendix~\ref{helplemma} (see Lemma~\ref{lhslemma1}~(2)). A quick calculation also shows that the probabilities sum to one:
\begin{align}
    \sum_i p(i|a, \lam)=f_a(\vec{v}_{s_xs_ys_z})+(1-f_a(\vec{v}_{s_xs_ys_z}))=1 \, .
\end{align}

Now we are in a position to show that
\begin{align}
    A^{1/2}_{i|a}=\int_{S_2} \mathrm{d}\lam \ p(i|a,\lam)\ G_{\lam} \, . \label{toshow}
\end{align}
We give the detailed proof in Appendix~\ref{identity} but sketch the main idea here. It is important to recognize that, the function $p(i|a, \lam)$ is the same for two different $\lam$ that lie in the same octant of the rotated frame $\vec{x}'$, $\vec{y}'$, $\vec{z}'$. Intuitively speaking, this leads to a coarse-graining of the measurement outcomes $\lam$ in each of these octants. These coarse-grained operators $G_{s_xs_ys_z}$ behave like a noisy measurement in the direction of the corresponding vector $\vec{v}_{s_xs_ys_z}$. More precisely, we calculate in Appendix~\ref{integrals} that:
\begin{align}
G_{s_xs_ys_z}:=&\int_{S_2|\sgn{(\lam \cdot \vec{k}')}=s_k} \mathrm{d}\lam \ G_{\lam} =\frac{\id}{8}+\frac{\vec{v}_{s_xs_ys_z}\cdot \vec{\sigma}}{16}\, .
\end{align}
With this definition, Eq.~\eqref{toshow} becomes:
\begin{align}
    A^{1/2}_{i|a}=\sum_{s_x,s_y,s_z=\pm 1} p(i|a,\lam) \ G_{s_xs_ys_z} \, .
\end{align}
Using the definition of $p(i|a,\lam)$ in Eq.~\eqref{condprob} and some algebra  (details in Appendix~\ref{identity}) this reduces to:
\begin{align}
    A^{1/2}_{i|a}=\sum_{s_x,s_y,s_z=\pm 1} p_i \ \Theta(\vec{a}_i\cdot \vec{v}_{s_xs_ys_z})\ G_{s_xs_ys_z} + \alpha_i \id \, .
    \label{identitymain}
\end{align}
In the end, we prove this identity by using a closely related geometric formula that decomposes $\vec{a}_i$ into the vectors $\vec{v}_{s_xs_ys_z}$ (see Appendix~\ref{helplemma}):
\begin{align}
    \sum_{s_x,s_y,s_z=\pm 1} \Theta(\vec{a}_i\cdot \vec{v}_{s_xs_ys_z})\ \vec{v}_{s_xs_ys_z}=4\  \vai \, .
\end{align}

The identity in Eq.~\eqref{identitymain} can be seen as the main idea of the protocol. We want to find a set of coarse-grained operators $G_{s_xs_ys_z}$ that can be used to decompose all the POVM elements $A^{1/2}_{i|a}$. The conditional probabilities $p(i|a, \lam)$ are then constructed according to this decomposition. The first term in $p(i|a, \lam)$, namely $p_i\ \Theta(\vec{v}_{s_xs_ys_z}\cdot \vec{a}_i)$ is the coefficient that comes from the decomposition of $A^{1/2}_{i|a}$ in terms of $G_{s_xs_ys_z}$. The second term in $p(i|a, \lam)$ is constructed to add the noise term $\alpha_i \id$.

To give an example, consider the blue vector in Fig.~\ref{illsicmain} for which $\vec{a}_1=(0,0,1)^T$ and $p_1=1/2$, hence $A^{1/2}_{1|a}=p_1 \big(\mathds{1} +  1/2\ \vec{a}_1 \cdot \vec{\sigma}\big)/2=\id/4+ \sigma_z/8$. It turns out that we can use the standard coordinate frame in which the cube vertices are simply $\vec{v}_{\pm \pm \pm}:=(\pm 1, \pm 1, \pm 1)^T$. Direct calculation shows that $p_1 \cdot \Theta(\vec{a}_1\cdot \vec{v}_{s_xs_ys_z})=1/2$ if $s_z=+1$ (and zero if $s_z=-1$) as well as $\alpha_1=0$. In addition, the coarse-grained operators become $G_{s_xs_ys_z}=\id/8 +(s_x \sigma_X+s_y \sigma_Y+s_z \sigma_Z)/16$. It is then easy to check that $1/2(G_{+++}+G_{+-+}+G_{-++}+G_{--+})=A^{1/2}_{1|a}$.

However, while the identity in Eq.~\eqref{identitymain} holds for any orthonormal frame, it can be translated into a protocol with well-defined probabilities only if $\sum_i p_i\cdot \Theta(\vec{v}_{s_xs_ys_z}\cdot \vec{a}_i)=f_a(\vec{v}_{s_x s_y s_z})\leq 1$ for all eight vertices of the cube. We show now that such a frame always exists. The proof has two steps. First, we show that for any such cube, it holds that:
\begin{align}
    \sum_{s_x, s_y, s_z =\pm 1} f_a(\vec{v}_{s_x s_y s_z})\leq 8 \, .\label{eqbound}
\end{align}
The second part of the proof uses a theorem by Hausel, Makai, and Szűcs \cite{hausel_makai_szűcs_2000} (see Theorem~1 in that reference) that applies to continuous real-valued functions on $S_2$ that have the additional property that $f(\vec{x})=f(-\vec{x})$. By using similar techniques as in Ref.~\cite{Renner2023}, we show in Appendix~\ref{propertiesfunction} that these conditions are indeed fulfilled. In their theorem, they show that there always exists a rotation of the cube such that the functional values coincide at all eight vertices of that cube. Hence, choosing the orthonormal frame according to that rotation, we obtain $f_a(\vec{v}_{s_x s_y s_z})=C$ for all $s_x, s_y, s_z \in \{+1,-1\}$. Combining this with the above bound in Eq.~\eqref{eqbound}, we get $8C\leq 8$ and therefore $f_a(\vec{v}_{s_x s_y s_z})\leq 1$ for that specific cube (see Appendix~\ref{propertiesfunction} for more details).

The theorem in Ref.~\cite{hausel_makai_szűcs_2000} is a special case of a family of so-called Knaster-type theorems. They state that for a given continuous real-valued function on the sphere, a certain configuration of points can always be rotated such that the functional values coincide at each of these points. Other interesting related results concerning $S_2$ are due to Dyson~\cite{Dyson1951}, Livesay~\cite{Livesay1954}, Floyd~\cite{Floyd1955}. Also, the well-known Borsuk–Ulam theorem is of this type~\cite{Borsuk_1933}.

We want to remark that we do not necessarily have to choose a cube in which all of these eight values coincide. It is only required that all of these eight values are smaller than one. Note that we do not give an explicit way to construct such a coordinate frame. However, in many cases, for instance, for POVMs with two or three outcomes, it turns out that an explicit construction can be found. We discuss this further in Appendix~\ref{appendixprl} and Appendix~\ref{specialcases} (see also there for further examples and more illustrations).

\textit{Local models for entangled quantum states.---} Now we apply the developed techniques to Bell nonlocality and quantum steering. Suppose Alice and Bob share a two-qubit Werner state \cite{Werner1989}:
\begin{align}
    \rho^\eta_{W}=\eta \ketbra{\Psi^-} + (1-\eta)\ \id/4 \, ,
\end{align}
where $\ket{\Psi^-}=\frac{1}{\sqrt{2}}(\ket{01}-\ket{10})$ denotes the two-qubit singlet. They can apply arbitrary local POVMs on their qubit. As before, we denote Alice's measurement operators with $A_{i|a}=p_i \ketbra{\vec{a}_i}=\big(\mathds{1} + \vec{a}_i \cdot \vec{\sigma}\big)/2$ (where $p_i\geq 0$, $|\vec{a}_i|=1$, $\sum_i p_i=2$, and $\sum_i p_i\ \vai=\vec{0}$). Similarly, Bob can perform an arbitrary POVM with elements $B_{j|b}$ that are defined analogously. Note, that Alice's and Bob's measurements are now completely arbitrary, i.e. they are not noisy. Instead, the entangled state is not pure but has a certain amount of white noise. The correlations when Alice and Bob apply local POVMs to this state become:
\begin{align}
    p(i,j|a,b)=&\tr[(A_{i|a}\otimes B_{j|b})\ \rho^\eta_{W}] \, .
\end{align}

We say that the state $\rho_{W}^\eta$ admits a local hidden variable model, if these correlations can be written as
\begin{align}
    p(i,j|a,b)=\int \mathrm{d}\lambda\ q(\lambda)\ p(i|a,\lambda) \ p(j|b,\lambda) \, ,
\end{align}
for some hidden variable $\lambda$ distributed according to $q(\lambda)$ and some conditional probabilities $p(i|a,\lambda)$ and $p(j|b,\lambda)$. If the state $\rho_{W}^\eta$ violates a Bell inequality, such a local hidden variable description cannot exist and we say that the state is nonlocal~\cite{Bell, Brunner_2014}. It is a fundamental question in Bell nonlocality, for which $\eta$ these correlations can violate a Bell inequality or admit a local hidden variable model. It is known, that two-qubit Werner states violate the CHSH inequality \cite{CHSH} for $\eta > 1/\sqrt{2} \approx 0.7071$. Vertesi showed that they violate another Bell inequality whenever $\eta > 0.7056$ \cite{Vertesi2008}.

On the other hand, Werner constructed in his seminal paper from 1989 a local model for all bipartite projective measurements if $\eta \leq 1/2$ albeit these states are entangled if $\eta>1/3$ \cite{Werner1989}. Later, this bound was improved by Acin, Toner, and Gisin, who showed that the state is local whenever $\eta \leq 1/K_G(3)$ \cite{AcinGrothen}. Here, $K_G(3)$ is the so-called Grothendieck constant of order three and the best current bound is by Designolle et al. $1.4367 \leq K_G(3) \leq 1.4546$ \cite{Designolle2023}. This implies that $\rho^{\eta}_W$ is local if $\eta\leq 0.6875$ and violates a Bell inequality if $\eta\geq 0.6961$. However, these local models only apply to projective measurements (where
$p_1=p_2=1$ and $\vec{a}_2=-\vec{a}_1$).

Considering general POVMs, Barrett found a local model for all POVMs whenever $\eta \leq 5/12$ \cite{Barrett2002}. Using a technique developed in Ref.~\cite{Hirsch2017, Oszmaniec2017}, the best bound is again by Ref.~\cite{Designolle2023} which shows that $\rho^{\eta}_W$ is local for all POVMs if $\eta\leq 0.4583$. Based on the connections made in Ref.~\cite{Quintino2014, Roope2014, Uola2015}, we can now show that whenever $\eta\leq 1/2$ we cannot violate any Bell inequality since all correlations can be described by the following local model.

Suppose Alice performs her measurement $\{A_{i|a}\}_i$ on the Werner state with $\eta=1/2$ and obtains outcome $i$. After doing so, Bob's qubit is precisely in the (unnormalized) post-measurement state:
\begin{align}
    \rho_B(i|a)=&\tr_A[(A_{i|a}\otimes \id) \rho^{1/2}_{W}]=p_i \big(\mathds{1} - \vec{a}_i \cdot \vec{\sigma}/2\big)/4 \, . \label{postmeasu}
\end{align}
It is now important to recognize that this state can be simulated with the same techniques as before, due to the duality of states and measurements. More precisely, consider the following protocol:
\begin{enumerate}
	\item  $\lam \in \mathbb{R}^3$ is a normalized vector, drawn randomly from the unit radius sphere $S_2$ (according to the Haar measure). Alice knows the vector $\lam$. Bob's system is in the pure qubit state $\rho_{\lam}={\big(\mathds{1} + \vec{\lambda}\cdot\vec{\sigma}\big)/2}$.
        \item Alice chooses her POVM with operators $A_{i|a}=p_i(\mathds{1} + \vec{a}_i \cdot \vec{\sigma})/2$. Now, she applies precisely the same steps as in the previous protocol for the given values of $p_i$, vectors $-\vai$ ("-" to account for the anticorrelations in the singlet), and $\lam$. Namely, she chooses a suitable frame and produces her outcome $i$ according to the conditional probabilities $p(i|-a,\lam)$ in Eq.~\eqref{condprob}. (The "-" in $-a$ accounts for the "-" in $-\vai$.)
	\item Bob chooses his POVM with elements $B_{j|b}$ and performs a quantum measurement on his state $\rho_{\lam}$.
\end{enumerate}

The distribution of the state $\rho_{\lam}$, namely $\frac{1}{8\pi}\big(\mathds{1} + \vec{\lambda}\cdot\vec{\sigma}\big)$, is the same expression as the one for the parent POVM in Eq.~\eqref{parentPOVM} (up to a factor of $2$ since states and measurements are normalized differently). Hence, if we sum over all the states where Alice outputs $i$, she samples precisely the state $p_i \big(\mathds{1} - \vec{a}_i \cdot \vec{\sigma}/2\big)/4$ (analog to $A^{1/2}_{i|a}=p_i(\mathds{1} + \vec{a}_i \cdot \vec{\sigma}/2)/2$ before). This matches exactly the expression in Eq.~\eqref{postmeasu}. Intuitively speaking, there is no difference for Bob's qubit if Alice performs the protocol above or performs the measurement on the actual Werner state for $\eta =1/2$. Therefore, when Bob applies his POVM, the resulting statistics are the same in both cases. Hence, the protocol above simulates the statistics of arbitrary POVMs applied to the state $\rho^{1/2}_{W}$ in a local way:\footnote{In the protocol, it seems that we need quantum resources to simulate the statistics. However, we can also assume that $\lam$ is known to both and then Bob can output $j$ with probability $p(j|b,\lam)=\tr[B_{j|b}\ \rho_{\lam}]$ using his knowledge of $\lam$ and his measurement operators $B_{j|b}$.}
\begin{align}
    \tr[(A_{i|a} \otimes B_{j|b})\ \rho_W^{1/2}]=\frac{1}{4\pi}\int_{S_2} \mathrm{d}\vec{\lambda}\ p(i|-a,\vec{\lambda}) \tr[B_{j|b}\ \rho_{\vec{\lambda}}] \, .
\end{align}

This model is even a so-called local hidden state model which implies that the state $\rho^{1/2}_{W}$ is not steerable \cite{EPR1935, Wiseman2007, SteeringReview2020}. In the most fundamental steering scenario, we consider two parties, Alice and Bob, that share an entangled quantum state. The question is, whether Alice can steer Bob's state by applying a measurement on her side. However, Bob wants to exclude the possibility that his system is prepared in a well-defined state that is known to Alice. Then, Alice could just use her knowledge of the "hidden state" to pretend to Bob that she can steer his state. However, in reality, they do not share any entanglement at all. This is precisely the case in the above protocol, proving that the state $\rho^{\eta}_W$ cannot demonstrate quantum steering whenever $\eta\leq 1/2$. This was known before for the restricted case of projective measurements $A_{\pm|a}=(\id \pm \vec{a}\cdot \vec{\sigma})/2$~\cite{Wiseman2007}. When general POVMs are considered, the best model so far is the one from Barrett~\cite{Barrett2002}, which was shown to be a local hidden state model by Quintino et al.~\cite{Quintino2015}. That model shows that $\rho_W^{\eta}$ cannot demonstrate steering if $\eta\leq 5/12$. Numerical evidence suggested that the same holds for all $\eta\leq 1/2$ \cite{Bavaresco2017, Chau2018, Nguyen2019}. Our model shows, that this is indeed the case.

On the other hand, if such a local hidden state model cannot exist we say that the state is steerable. It is known that the two-qubit Werner state can demonstrate steering whenever $\eta> 1/2$ \cite{Wiseman2007}. Therefore, the bound of $\eta= 1/2$ is tight. Due to the connection between steering and joint measureability \cite{Roope2014, Quintino2014, Uola2015}, $\eta= 1/2$ is also tight for the joint measureability problem, ensuring the optimality of our construction.

\textit{Conclusion.---} 
In this work, we provided tight bounds on how much white noise a measurement device can tolerate before all qubit measurements become jointly measurable. We considered the most general set of measurements (POVMs) and applied our techniques to quantum steering and Bell nonlocality. Exploiting the connection between joint measurability and steering \cite{Roope2014, Quintino2014, Uola2015}, we found a tight local hidden state model for two-qubit Werner states of visibility $\eta=1/2$. This solves Problem 39 on the page of Open quantum problems \cite{openproblem} (see also Ref.~\cite{Werner2014a}) and Conjecture~1 of Ref.~\cite{Chau2018}. An important direction for further research is the generalization to higher dimensional systems~\cite{Chau2020, Almeida2007}.

\section*{Acknowledgments}
Most importantly, I acknowledge Marco Túlio Quintino for important discussions and for introducing the problem to me. Furthermore, I acknowledge Haggai Nuchi for correspondence about the Knaster-type theorems and an anonymous referee for very useful suggestions. This research was funded in whole, or in part, by the Austrian Science Fund (FWF) through BeyondC with Grant-DOI 10.55776/F71.\\

\textit{Note added:} At the very last stage of this work, we became aware of the work by Yujie Zhang and Eric Chitambar~\cite{Zhang2023} that proves the same results with a different approach. Both works appeared the same week on arXiv and are published back to back in the same journal.

\section{Appendix: A note on the non-constructive nature of the protocol}\label{appendixprl}
Here, we would like to provide some additional information about the non-constructive nature of our protocol. We stress again that the theorem of Hausel, Makai, and Szűcs \cite{hausel_makai_szűcs_2000} only implies that a suitable coordinate frame exists but does not imply how to find one. However, in some cases, we can explicitly find a frame. 

Consider for instance the important special case of a POVM with only two outcomes which corresponds to a projective measurement. In that case, we have $p_1=p_2=1$ and $\vec{a}_2=-\vec{a}_1$, hence $A^{1/2}_{1|a}=(\id+ \vec{a}_1\cdot \vec{\sigma}/2)/2$ and $A^{1/2}_{2|a}=(\id- \vec{a}_1\cdot \vec{\sigma}/2)/2$. We can express the function $f_a(\vec{x})$ as $f_a(\vec{x})=\Theta(\vec{x}\cdot \vec{a}_{1})+\Theta(-\vec{x}\cdot \vec{a}_{1})=|\vec{x}\cdot \vec{a}_{1}|$. To find a suitable frame, we can choose the $x'$-axis to be aligned with $\vec{a}_{1/2}$, while the $y'$- and $z'$-axis are orthogonal to $\vec{a}_{1}$. In this way, $\vec{x}'=\vec{a}_{1}$ and direct calculation shows that $f_a(\vec{v}_{s_xs_ys_z})=1$ for all eight vertices $\vec{v}_{s_xs_ys_z}=s_x\vec{x}'+s_y\vec{y}'+s_z\vec{z}'$ as required. In addition, note that $\Theta(\vec{v}_{s_xs_ys_z}\cdot \vec{a}_{1})=1$ if $s_x=+1$ and $\Theta(\vec{v}_{s_xs_ys_z}\cdot \vec{a}_{1})=0$ if $s_x=-1$. Therefore, $\alpha_1=\alpha_2=0$ and the conditional probabilities translate precisely to $p(1|a, \lam)=1$ if $\lam\cdot \vec{a}_1\geq 0$ and $p(1|a, \lam)=0$ if $\lam\cdot \vec{a}_1< 0$ (and the analog expression for $i=2$). See Fig.~\ref{illproj} for an illustration.

The choice of the frame is unique up to an arbitrary rotation around the $x'$-axis (and a relabelling of the axes). To see this, note that the angle $\alpha$ between $\vec{a}_{1/2}$ and each cube vertex $\vec{v}_{s_xs_ys_z}$ must be at least $\alpha \geq \cos^{-1}{(1/\sqrt{3})}$ since $|\vec{a}_{1/2}|=1$, $|\vec{v}_{s_xs_ys_z}|=\sqrt{3}$ and $f_a(\vec{v}_{s_xs_ys_z})=|\vec{v}_{s_xs_ys_z}\cdot \vec{a}_{1}|=|\vec{v}_{s_xs_ys_z}| \cdot |\vec{a}_{1}| \cdot \cos{(\alpha)}$. Geometrically, this defines the cube uniquely up to a rotation around $\vec{a}_{1}$. Note, however, that such a rotation would not change the conditional probabilities $p(i|a,\lam)$ in the end. It is worth pointing out, that this construction becomes equivalent to the one of Werner~\cite{Werner1989}, which is known to be a tight local hidden state model for projective measurements~\cite{Wiseman2007} (and therefore also tight for the problem of joint measurability due to the close connection of these two fields~\cite{Quintino2014, Roope2014, Uola2015}).

\begin{figure}[ht!]
    \centering
    \includegraphics[width=1.0\columnwidth]{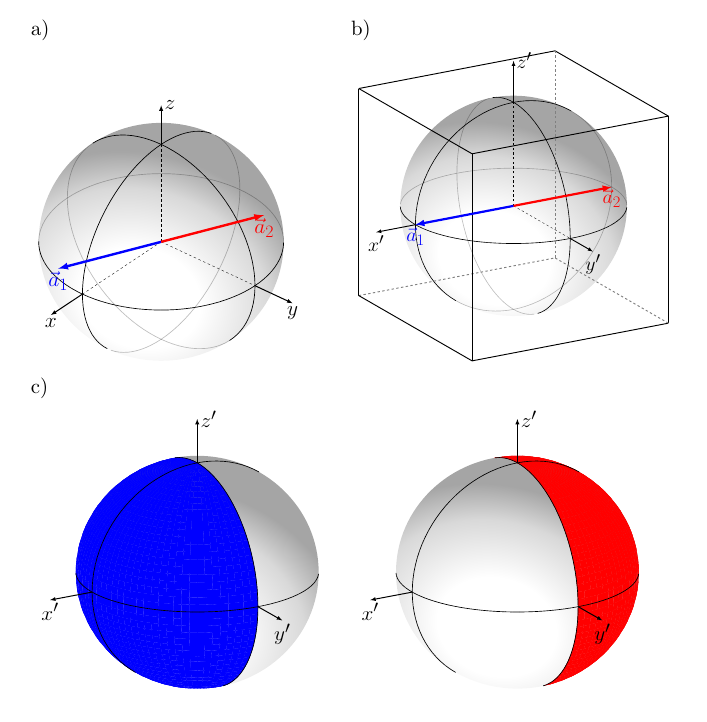}
    \caption{Construction for the two-outcome POVM with operators $A^{1/2}_{i|a}=(\id+ \vec{a}_i\cdot \vec{\sigma}/2)/2$ (with $\vec{a}_2=-\vec{a}_1$): \textbf{a)} Here, $\vec{a}_1$ can be an arbitrary direction in the Bloch sphere. \textbf{b)} We can choose the rotated frame such that the $x'$-axis is aligned with $\vec{a}_1$. We also show the corresponding cube here. \textbf{c)} The conditional probabilities $p(i|a,\lam)$ reduce precisely to $p(1|a,\lam)=1$ if $\lam\cdot \vec{a}_1\geq 0$ and $p(1|a, \lam)=0$ if $\lam\cdot \vec{a}_1< 0$ as indicated with the two colors. Hence, if the outcome $\lam$ of the parent POVM lies in the hemisphere centered around $\vec{a}_1$ (blue region) the outcome is always $i=1$ and if it lies in the hemisphere centered around $\vec{a}_2$ (red region), the outcome will be $i=2$.} \label{illproj}
\end{figure}

It turns out that for the case of three outcome POVMs, we can also construct a suitable coordinate frame without relying on the theorem of Hausel, Makai, and Szűcs \cite{hausel_makai_szűcs_2000} but only on the intermediate value theorem for continuous functions (see Appendix~\ref{specialcases}). For general POVMs, we want to point out that a suitable frame is computationally easy to find for many cases. For instance, we can parametrize a rotation by its three Euler angles. When we discretize the three angles into equally spaced values, we can search through many possible rotations and calculate the functional values for the corresponding cube. If we find a cube, for which all of these eight values are smaller than one, we have found a suitable frame. We provide a MATLAB code for this simple algorithm via GITHUB~\cite{GIthubMJR}. It turns out that even this brute-force method finds a suitable frame for most POVMs almost immediately. (However, more sophisticated algorithms, are likely to perform even better.)

We did some numerical simulations with random POVMs. For that, we generate random points on the sphere $\vec{a}_i\in S_2$ and find $p_i$ by solving $\sum_i p_i \cdot \vai=\vec{0}$ and $\sum_i p_i=2$. Then we use our algorithm to find a suitable frame. These numerical simulations strongly suggest that it is the hardest to find a frame if all directions are almost collinear ($|\vec{a}_i\cdot \vec{a}_j|\approx |\vec{a}_i|\cdot |\vec{a}_j|$ for all pairs $i,j$). Note that, the two outcome POVMs from above are precisely of that form. However, even in these cases, a frame was always found in which the largest of the eight values $f(\vec{v}_{s_xs_ys_z})$ is only slightly larger than 1 and this value can be further decreased by further discretizing the Euler angles. There is an intuitive explanation for this effect. For the simulation of a given POVM, it is always advantageous if a given $\lam$ is mapped to an outcome $i$ that is close, meaning that the angle between $\lam$ and $\vai$ is small. Consider for instance the case of a projective measurement discussed before (or a POVM with almost collinear vectors $\vai$). If $\lam$ is (almost) orthogonal to $\vec{a}_{1/2}$, the outcome of the parent measurement $\lam$ is (almost) uncorrelated to the $\vec{a}_{1/2}$ measurement but it still has to be mapped to either $\vec{a}_{1}$ or $\vec{a}_{2}$.

Contrary to that, for a POVM with more outcomes $\vai$ spread over the Bloch sphere (like the symmetric, informationally complete (SIC) POVM in Fig.~\ref{illsicmain}), there are more options a given $\lam$ can be mapped to. Roughly speaking it is then more likely to find a measurement outcome $\vai$ that is highly correlated with the actual measurement outcome of the parent POVM $\lam$. Based on this intuition, it is reasonable to expect that these POVMs are easier to simulate. In our construction, this expresses itself in the fact that for these POVMs many different coordinate frames are suitable, and therefore several different simulations for such a POVM and $\eta=1/2$ exist. We can even prove, that for the case of the four-outcome SIC-POVM \cite{SIC}, any rotation can be chosen (see Appendix~\ref{appsicevery}). On the contrary, $\eta=1/2$ is known to be tight for the special case of two-outcome POVMs ~\cite{Wiseman2007}, and therefore only very particular coordinate frames are possible (similar for collinear POVMs).

Note also, that we do not exclude the possibility that for certain POVMs better constructions with $\eta>1/2$ exist. For instance, SIC-POVMs are by definition very symmetric and one would expect that a symmetric model gives an even better bound $\eta>1/2$ (e.g., one can map $\lam$ to the closest outcome $\vai$ of the SIC-POVM). However, in this work, we are merely concerned with finding one construction that works for all POVMs and $\eta=1/2$. Hence, it is more important for our approach to recover the hemisphere construction of Fig.~\ref{illproj} (which is known to be tight for projective measurements) than to maintain the symmetry of a given POVM.


\nocite{apsrev42Control} 
\bibliographystyle{0_MTQ_apsrev4-2_corrected}
\bibliography{bib.bib}

@CONTROL{apsrev42Control,author="",editor="1",pages="1",title="0",year="0"}

@misc{SIC,
author = {Wikipedia},
note = "(accessed 21/09/2023)",
year = {2022},
title = {https://en.wikipedia.org/wiki/SIC-POVM},
   url={https://en.wikipedia.org/wiki/SIC-POVM}
}

@misc{openproblem,
note = "(accessed 12/07/2023)",
year = {2023},
title = {Open Quantum Problems, 
IQOQI Vienna: Steering bound for qubits and POVMs},
   url={https://oqp.iqoqi.oeaw.ac.at/steering-bound-for-qubits-and-povms}
}

@ARTICLE{Heisenberg1927,
       author = {Werner Heisenberg},
        title = "{Über den anschaulichen Inhalt der quantentheoretischen Kinematik und Mechanik}",
  journal={Zeitschrift für Physik},
         year = 1927,
        volume = {43},
    month = {Mar},
        pages = {172–198},
          doi = {https://doi.org/10.1007/BF01397280},
}

@article{Bell,
  title = {On the Einstein Podolsky Rosen paradox},
  author = {Bell, J. S.},
  journal = {Physics Physique Fizika},
  volume = {1},
  issue = {3},
  pages = {195--200},
  numpages = {6},
  year = {1964},
  month = {Nov},
  publisher = {American Physical Society},
  doi = {10.1103/PhysicsPhysiqueFizika.1.195},
  url = {https://link.aps.org/doi/10.1103/PhysicsPhysiqueFizika.1.195}
}

@ARTICLE{Barrett2002,
       author = {{Barrett}, Jonathan},
        title = "{Nonsequential positive-operator-valued measurements on entangled mixed states do not always violate a Bell inequality}",
      journal = {\pra},
         year = 2002,
        month = apr,
       volume = {65},
       number = {4},
          eid = {042302},
        pages = {042302},
          doi = {10.1103/PhysRevA.65.042302},
archivePrefix = {arXiv},
       eprint = {quant-ph/0107045},
 primaryClass = {quant-ph},
       adsurl = {https://ui.adsabs.harvard.edu/abs/2002PhRvA..65d2302B}
}

@article{DAriano2005,
	doi = {10.1088/0305-4470/38/26/010},
	url = {https://doi.org/10.1088/0305-4470/38/26/010},
	year = 2005,
	month = {jun},
	publisher = {{IOP} Publishing},
	volume = {38},
	number = {26},
	pages = {5979--5991},
	author = {Giacomo Mauro D{\textquotesingle}Ariano and Paoloplacido Lo Presti and Paolo Perinotti},
	title = {Classical randomness in quantum measurements},
	journal = {Journal of Physics A: Mathematical and General},
archivePrefix = {arXiv},
       eprint = {quant-ph/0408115},
 primaryClass = {quant-ph}
}

@ARTICLE{AcinGrothen,
       author = {{Ac{\'\i}n}, Antonio and {Gisin}, Nicolas and {Toner}, Benjamin},
        title = "{Grothendieck's constant and local models for noisy entangled quantum states}",
      journal = {\pra},
         year = 2006,
        month = jun,
       volume = {73},
       number = {6},
          eid = {062105},
        pages = {062105},
          doi = {10.1103/PhysRevA.73.062105},
archivePrefix = {arXiv},
       eprint = {quant-ph/0606138},
 primaryClass = {quant-ph},
       adsurl = {https://ui.adsabs.harvard.edu/abs/2006PhRvA..73f2105A}
}

@ARTICLE{Vertesi2008,
       author = {{V{\'e}rtesi}, T.},
        title = "{More efficient Bell inequalities for Werner states}",
      journal = {\pra},
         year = 2008,
        month = sep,
       volume = {78},
       number = {3},
          eid = {032112},
        pages = {032112},
          doi = {10.1103/PhysRevA.78.032112},
archivePrefix = {arXiv},
       eprint = {0806.0096},
 primaryClass = {quant-ph},
       adsurl = {https://ui.adsabs.harvard.edu/abs/2008PhRvA..78c2112V}
}

@ARTICLE{Oszmaniec2017,
       author = {{Oszmaniec}, Micha{\l} and {Guerini}, Leonardo and {Wittek}, Peter and {Ac{\'\i}n}, Antonio},
        title = "{Simulating Positive-Operator-Valued Measures with Projective Measurements}",
      journal = {\prl},
         year = 2017,
        month = nov,
       volume = {119},
       number = {19},
          eid = {190501},
        pages = {190501},
          doi = {10.1103/PhysRevLett.119.190501},
archivePrefix = {arXiv},
       eprint = {1609.06139},
 primaryClass = {quant-ph},
       adsurl = {https://ui.adsabs.harvard.edu/abs/2017PhRvL.119s0501O}
}

@ARTICLE{Quintino2015,
       author = {{Quintino}, Marco T{\'u}lio and {V{\'e}rtesi}, Tam{\'a}s and {Cavalcanti}, Daniel and {Augusiak}, Remigiusz and {Demianowicz}, Maciej and {Ac{\'\i}n}, Antonio and {Brunner}, Nicolas},
        title = "{Inequivalence of entanglement, steering, and Bell nonlocality for general measurements}",
      journal = {\pra},
         year = 2015,
        month = sep,
       volume = {92},
       number = {3},
          eid = {032107},
        pages = {032107},
          doi = {10.1103/PhysRevA.92.032107},
archivePrefix = {arXiv},
       eprint = {1501.03332},
 primaryClass = {quant-ph},
       adsurl = {https://ui.adsabs.harvard.edu/abs/2015PhRvA..92c2107Q}
}

@ARTICLE{Hirsch2017,
       author = {{Hirsch}, Flavien and {Quintino}, Marco T{\'u}lio and {V{\'e}rtesi}, Tam{\'a}s and {Navascu{\'e}s}, Miguel and {Brunner}, Nicolas},
        title = "{Better local hidden variable models for two-qubit Werner states and an upper bound on the Grothendieck constant $K_G(3)$}",
      journal = {Quantum},
         year = 2017,
        month = apr,
       volume = {1},
        pages = {3},
          doi = {10.22331/q-2017-04-25-3},
archivePrefix = {arXiv},
       eprint = {1609.06114},
 primaryClass = {quant-ph},
       adsurl = {https://ui.adsabs.harvard.edu/abs/2017Quant...1....3H}
}

@ARTICLE{Designolle2023,
       author = {{Designolle}, S{\'e}bastien and {Iommazzo}, Gabriele and {Besan{\c{c}}on}, Mathieu and {Knebel}, Sebastian and {Gel{\ss}}, Patrick and {Pokutta}, Sebastian},
        title = "{Improved local models and new Bell inequalities via Frank-Wolfe algorithms}",
      journal = {Physical Review Research},
         year = 2023,
        month = oct,
       volume = {5},
       number = {4},
          eid = {043059},
        pages = {043059},
          doi = {10.1103/PhysRevResearch.5.043059},
archivePrefix = {arXiv},
       eprint = {2302.04721},
 primaryClass = {quant-ph},
       adsurl = {https://ui.adsabs.harvard.edu/abs/2023PhRvR...5d3059D}
}

@misc{GIthubMJR,
   author = {Renner, Martin J.},
   title = {https://github.com/MartinJRenner},
   year = {2024},
    url = {https://github.com/MartinJRenner/CompatibilityQubitPOVMs.git}
}

@article{Brunner_2014,
	doi = {10.1103/revmodphys.86.419},
	url = {https://doi.org/10.1103%2Frevmodphys.86.419},
	year = 2014,
	month = {Apr},
	publisher = {American Physical Society ({APS})},
	volume = {86},
	number = {2},
	pages = {419--478},
	author = {Nicolas Brunner and Daniel Cavalcanti and Stefano Pironio and Valerio Scarani and Stephanie Wehner},
	title = {Bell nonlocality},
	journal = {Reviews of Modern Physics},
archivePrefix = {arXiv},
       eprint = {1303.2849},
 primaryClass = {quant-ph}
}

@ARTICLE{FrenkelWeiner2015,
       author = {{Frenkel}, P{\'e}ter E. and {Weiner}, Mih{\'a}ly},
        title = "{Classical Information Storage in an n-Level Quantum System}",
      journal = {Communications in Mathematical Physics},
         year = 2015,
        month = Dec,
       volume = {340},
       number = {2},
        pages = {563-574},
          doi = {10.1007/s00220-015-2463-0},
archivePrefix = {arXiv},
       eprint = {1304.5723},
 primaryClass = {cs.IT},
       adsurl = {https://ui.adsabs.harvard.edu/abs/2015CMaPh.340..563F}
}

@article{Livesay1954,
 ISSN = {0003486X},
 URL = {http://www.jstor.org/stable/1969689},
 author = {George R. Livesay},
 journal = {Annals of Mathematics},
 number = {2},
 pages = {227--229},
 publisher = {Annals of Mathematics},
 title = {On a Theorem of F. J. Dyson},
 urldate = {2023-06-13},
 volume = {59},
 year = {1954}
}

@article{Dyson1951,
 ISSN = {0003486X},
 URL = {http://www.jstor.org/stable/1969487},
 author = {F. J. Dyson},
 journal = {Annals of Mathematics},
 number = {3},
 pages = {534--536},
 publisher = {Annals of Mathematics},
 title = {Continuous Functions Defined on Spheres},
 urldate = {2023-06-13},
 volume = {54},
 year = {1951}
}

@article{Floyd1955,
 author = {Edwin E. Floyd},
 journal = {Proc. Amer. Math. Soc.},
 pages = {957--959},
 title = {Real-valued mappings of spheres},
 volume = {6},
 year = {1955}
}

@article{hausel_makai_szűcs_2000,
    title={Inscribing cubes and covering by rhombic dodecahedra via equivariant topology},
    volume={47},
    DOI={10.1112/S0025579300015965},
    number={1-2},
    journal={Mathematika}, 
    publisher={London Mathematical Society}, 
    author={Hausel, T. and Makai, E. and Szűcs, A.},
    year={2000},
    pages={371–397},
archivePrefix = {arXiv},
       eprint = {math/9906066},
 primaryClass = {math.MG},
}

@article{Borsuk_1933, title={Drei Sätze über die n-dimensionale euklidische Sphäre}, volume={20}, number={1}, journal={Fundamenta Mathematicae}, publisher={Polska Akademia Nauk. Instytut Matematyczny PAN}, author={Borsuk, Karol}, year={1933}, pages={177–190} }

@article{Compatibilityreview202,
  title = {Colloquium: Incompatible measurements in quantum information science},
  author = {G\"uhne, Otfried and Haapasalo, Erkka and Kraft, Tristan and Pellonp\"a\"a, Juha-Pekka and Uola, Roope},
  journal = {Rev. Mod. Phys.},
  volume = {95},
  issue = {1},
  pages = {011003},
  numpages = {25},
  year = {2023},
  month = {Feb},
  publisher = {American Physical Society},
  doi = {10.1103/RevModPhys.95.011003},
  url = {https://link.aps.org/doi/10.1103/RevModPhys.95.011003},
archivePrefix = {arXiv},
       eprint = {2112.06784},
 primaryClass = {quant-ph},
}

@article{Nguyen2019,
	author = {H. Chau Nguyen and Huy-Viet Nguyen and Otfried Gühne},
	title = {Geometry of Einstein-Podolsky-Rosen Correlations},
	year = 2019,
	month = {jun},
	journal = {\prl},
	volume = {122},  
	number = {24},  
          eid = {240401},
        pages = {240401},
	publisher = {American Physical Society ({APS})},
	doi = {10.1103/physrevlett.122.240401},  
archivePrefix = {arXiv},
       eprint = {1808.09349},
 primaryClass = {quant-ph}
}

@ARTICLE{SteeringReview2020,
       author = {{Uola}, Roope and {Costa}, Ana C.~S. and {Nguyen}, H. Chau and {G{\"u}hne}, Otfried},
        title = "{Quantum steering}",
      journal = {Reviews of Modern Physics},
         year = 2020,
        month = jan,
       volume = {92},
       number = {1},
          eid = {015001},
        pages = {015001},
          doi = {10.1103/RevModPhys.92.015001},
archivePrefix = {arXiv},
       eprint = {1903.06663},
 primaryClass = {quant-ph},
       adsurl = {https://ui.adsabs.harvard.edu/abs/2020RvMP...92a5001U}
}

@article{Werner2014a,
doi = {10.1088/1751-8113/47/42/424008},
url = {https://dx.doi.org/10.1088/1751-8113/47/42/424008},
year = {2014},
month = {oct},
publisher = {IOP Publishing},
volume = {47},
number = {42},
pages = {424008},
author = {R F Werner},
title = {Steering, or maybe why Einstein did not go all the way to Bells argument},
journal = {Journal of Physics A: Mathematical and Theoretical},
}

@article{Werner1989,
  title = {Quantum states with Einstein-Podolsky-Rosen correlations admitting a hidden-variable model},
  author = {Werner, Reinhard F.},
  journal = {Phys. Rev. A},
  volume = {40},
  issue = {8},
  pages = {4277--4281},
  numpages = {0},
  year = {1989},
  month = {Oct},
  publisher = {American Physical Society},
  doi = {10.1103/PhysRevA.40.4277},
  url = {https://link.aps.org/doi/10.1103/PhysRevA.40.4277}
}

@ARTICLE{Wiseman2007,
       author = {{Wiseman}, H.~M. and {Jones}, S.~J. and {Doherty}, A.~C.},
        title = "{Steering, Entanglement, Nonlocality, and the Einstein-Podolsky-Rosen Paradox}",
      journal = {\prl},
         year = 2007,
        month = apr,
       volume = {98},
       number = {14},
          eid = {140402},
        pages = {140402},
          doi = {10.1103/PhysRevLett.98.140402},
archivePrefix = {arXiv},
       eprint = {quant-ph/0612147},
 primaryClass = {quant-ph},
       adsurl = {https://ui.adsabs.harvard.edu/abs/2007PhRvL..98n0402W}
}

@ARTICLE{Roope2014,
       author = {{Uola}, Roope and {Moroder}, Tobias and {G{\"u}hne}, Otfried},
        title = "{Joint Measurability of Generalized Measurements Implies Classicality}",
      journal = {\prl},
         year = 2014,
        month = oct,
       volume = {113},
       number = {16},
          eid = {160403},
        pages = {160403},
          doi = {10.1103/PhysRevLett.113.160403},
archivePrefix = {arXiv},
       eprint = {1407.2224},
 primaryClass = {quant-ph},
       adsurl = {https://ui.adsabs.harvard.edu/abs/2014PhRvL.113p0403U}
}

@ARTICLE{Bavaresco2017,
       author = {{Bavaresco}, Jessica and {Quintino}, Marco T{\'u}lio and {Guerini}, Leonardo and {Maciel}, Thiago O. and {Cavalcanti}, Daniel and {Cunha}, Marcelo Terra},
        title = "{Most incompatible measurements for robust steering tests}",
      journal = {\pra},
         year = 2017,
        month = aug,
       volume = {96},
       number = {2},
          eid = {022110},
        pages = {022110},
          doi = {10.1103/PhysRevA.96.022110},
archivePrefix = {arXiv},
       eprint = {1704.02994},
 primaryClass = {quant-ph},
       adsurl = {https://ui.adsabs.harvard.edu/abs/2017PhRvA..96b2110B}
}

@ARTICLE{Quintino2014,
       author = {{Quintino}, Marco T{\'u}lio and {V{\'e}rtesi}, Tam{\'a}s and {Brunner}, Nicolas},
        title = "{Joint Measurability, Einstein-Podolsky-Rosen Steering, and Bell Nonlocality}",
      journal = {\prl},
         year = 2014,
        month = oct,
       volume = {113},
       number = {16},
          eid = {160402},
        pages = {160402},
          doi = {10.1103/PhysRevLett.113.160402},
archivePrefix = {arXiv},
       eprint = {1406.6976},
 primaryClass = {quant-ph},
       adsurl = {https://ui.adsabs.harvard.edu/abs/2014PhRvL.113p0402Q}
}

@ARTICLE{Uola2015,
       author = {{Uola}, Roope and {Budroni}, Costantino and {G{\"u}hne}, Otfried and {Pellonp{\"a}{\"a}}, Juha-Pekka},
        title = "{One-to-One Mapping between Steering and Joint Measurability Problems}",
      journal = {\prl},
         year = 2015,
        month = dec,
       volume = {115},
       number = {23},
          eid = {230402},
        pages = {230402},
          doi = {10.1103/PhysRevLett.115.230402},
archivePrefix = {arXiv},
       eprint = {1507.08633},
 primaryClass = {quant-ph},
       adsurl = {https://ui.adsabs.harvard.edu/abs/2015PhRvL.115w0402U}
}

@ARTICLE{Almeida2007,
       author = {{Almeida}, Mafalda L. and {Pironio}, Stefano and {Barrett}, Jonathan and {T{\'o}th}, G{\'e}za and {Ac{\'\i}n}, Antonio},
        title = "{Noise Robustness of the Nonlocality of Entangled Quantum States}",
      journal = {\prl},
         year = 2007,
        month = jul,
       volume = {99},
       number = {4},
          eid = {040403},
        pages = {040403},
          doi = {10.1103/PhysRevLett.99.040403},
archivePrefix = {arXiv},
       eprint = {quant-ph/0703018},
 primaryClass = {quant-ph},
       adsurl = {https://ui.adsabs.harvard.edu/abs/2007PhRvL..99d0403A}
}

@article{Busch1986,
  title = {Unsharp reality and joint measurements for spin observables},
  author = {Busch, Paul},
  journal = {Phys. Rev. D},
  volume = {33},
  issue = {8},
  pages = {2253--2261},
  numpages = {0},
  year = {1986},
  month = {Apr},
  publisher = {American Physical Society},
  doi = {10.1103/PhysRevD.33.2253},
  url = {https://link.aps.org/doi/10.1103/PhysRevD.33.2253}
}

@ARTICLE{Wolf2009,
       author = {{Wolf}, Michael M. and {Perez-Garcia}, David and {Fernandez}, Carlos},
        title = "{Measurements Incompatible in Quantum Theory Cannot Be Measured Jointly in Any Other No-Signaling Theory}",
      journal = {\prl},
         year = 2009,
        month = dec,
       volume = {103},
       number = {23},
          eid = {230402},
        pages = {230402},
          doi = {10.1103/PhysRevLett.103.230402},
archivePrefix = {arXiv},
       eprint = {0905.2998},
 primaryClass = {quant-ph},
       adsurl = {https://ui.adsabs.harvard.edu/abs/2009PhRvL.103w0402W}
}

@article{Fine1982,
  title = {Hidden Variables, Joint Probability, and the Bell Inequalities},
  author = {Fine, Arthur},
  journal = {Phys. Rev. Lett.},
  volume = {48},
  issue = {5},
  pages = {291--295},
  numpages = {0},
  year = {1982},
  month = {Feb},
  publisher = {American Physical Society},
  doi = {10.1103/PhysRevLett.48.291},
  url = {https://link.aps.org/doi/10.1103/PhysRevLett.48.291}
}

@article{Fine1982b,
    author = {Fine, Arthur},
    title = "{Joint distributions, quantum correlations, and commuting observables}",
    journal = {Journal of Mathematical Physics},
    volume = {23},
    number = {7},
    pages = {1306-1310},
    year = {1982},
    month = {07},
    issn = {0022-2488},
    doi = {10.1063/1.525514},
    url = {https://doi.org/10.1063/1.525514},
}

@misc{footnote,
title= "{Note that, in contrary to all other POVMs in this work, the parent POVM has a continuous set of outcomes $\lam$ and needs to satisfy $\int_{S_2} \mathrm{d}\lam \ G_{\lam}=\id$ as well as $G_{\vec{\lambda}}\geq 0$.}"}

@article{EPR1935,
  title = {Can Quantum-Mechanical Description of Physical Reality Be Considered Complete?},
  author = {Einstein, A. and Podolsky, B. and Rosen, N.},
  journal = {Phys. Rev.},
  volume = {47},
  issue = {10},
  pages = {777--780},
  numpages = {0},
  year = {1935},
  month = {May},
  publisher = {American Physical Society},
  doi = {10.1103/PhysRev.47.777},
  url = {https://link.aps.org/doi/10.1103/PhysRev.47.777}
}

@ARTICLE{Saha2023,
       author = {{Saha}, Debashis and {Das}, Debarshi and {Das}, Arun Kumar and {Bhattacharya}, Bihalan and {Majumdar}, A.~S.},
        title = "{Measurement incompatibility and quantum advantage in communication}",
      journal = {\pra},
         year = 2023,
        month = jun,
       volume = {107},
       number = {6},
          eid = {062210},
        pages = {062210},
          doi = {10.1103/PhysRevA.107.062210},
archivePrefix = {arXiv},
       eprint = {2209.14582},
 primaryClass = {quant-ph},
       adsurl = {https://ui.adsabs.harvard.edu/abs/2023PhRvA.107f2210S}
}

@article{CHSH,
  title = {Proposed Experiment to Test Local Hidden-Variable Theories},
  author = {Clauser, John F. and Horne, Michael A. and Shimony, Abner and Holt, Richard A.},
  journal = {Phys. Rev. Lett.},
  volume = {23},
  issue = {15},
  pages = {880--884},
  numpages = {0},
  year = {1969},
  month = {Oct},
  publisher = {American Physical Society},
  doi = {10.1103/PhysRevLett.23.880},
  url = {https://link.aps.org/doi/10.1103/PhysRevLett.23.880}
}

@ARTICLE{Carmeli2020,
       author = {{Carmeli}, Claudio and {Heinosaari}, Teiko and {Toigo}, Alessandro},
        title = "{Quantum random access codes and incompatibility of measurements}",
      journal = {EPL (Europhysics Letters)},
         year = 2020,
        month = jun,
       volume = {130},
       number = {5},
          eid = {50001},
        pages = {50001},
          doi = {10.1209/0295-5075/130/50001},
archivePrefix = {arXiv},
       eprint = {1911.04360},
 primaryClass = {quant-ph},
       adsurl = {https://ui.adsabs.harvard.edu/abs/2020EL....13050001C}
}

@ARTICLE{Carmeli2019,
       author = {{Carmeli}, Claudio and {Heinosaari}, Teiko and {Toigo}, Alessandro},
        title = "{Quantum Incompatibility Witnesses}",
      journal = {\prl},
         year = 2019,
        month = apr,
       volume = {122},
       number = {13},
          eid = {130402},
        pages = {130402},
          doi = {10.1103/PhysRevLett.122.130402},
archivePrefix = {arXiv},
       eprint = {1812.02985},
 primaryClass = {quant-ph},
       adsurl = {https://ui.adsabs.harvard.edu/abs/2019PhRvL.122m0402C}
}

@ARTICLE{Chau2018,
       author = {{Chau Nguyen}, H. and {Milne}, Antony and {Vu}, Thanh and {Jevtic}, Sania},
        title = "{Quantum steering with positive operator valued measures}",
      journal = {Journal of Physics A Mathematical General},
         year = 2018,
        month = aug,
       volume = {51},
       number = {35},
          eid = {355302},
        pages = {355302},
          doi = {10.1088/1751-8121/aad115},
archivePrefix = {arXiv},
       eprint = {1706.08166},
 primaryClass = {quant-ph},
       adsurl = {https://ui.adsabs.harvard.edu/abs/2018JPhA...51I5302C}
}

@article{Chau2020,  
	author = {H. Chau Nguyen and Otfried Gühne},  
	title = {Some Quantum Measurements with Three Outcomes Can Reveal Nonclassicality where All Two-Outcome Measurements Fail to Do So},  
	journal = {\prl},
	year = 2020,
	month = {dec},    
	volume = {125},  
	number = {23},
          eid = {230402},
        pages = {230402},
	doi = {10.1103/physrevlett.125.230402},  
archivePrefix = {arXiv},
       eprint = {2001.03514},
 primaryClass = {quant-ph},
}

@ARTICLE{He2013QKD,
       author = {{He}, Q.~Y. and {Reid}, M.~D.},
        title = "{Genuine Multipartite Einstein-Podolsky-Rosen Steering}",
      journal = {\prl},
         year = 2013,
        month = dec,
       volume = {111},
       number = {25},
          eid = {250403},
        pages = {250403},
          doi = {10.1103/PhysRevLett.111.250403},
archivePrefix = {arXiv},
       eprint = {1212.2270},
 primaryClass = {quant-ph},
       adsurl = {https://ui.adsabs.harvard.edu/abs/2013PhRvL.111y0403H}
}

@ARTICLE{IncompatibilityReview2015,
       author = {{Heinosaari}, Teiko and {Miyadera}, Takayuki and {Ziman}, Mario},
        title = "{An Invitation to Quantum Incompatibility}",
      journal = {arXiv e-prints},
         year = 2015,
        month = nov,
          eid = {arXiv:1511.07548},
        pages = {arXiv:1511.07548},
          doi = {10.48550/arXiv.1511.07548},
archivePrefix = {arXiv},
       eprint = {1511.07548},
 primaryClass = {quant-ph},
       adsurl = {https://ui.adsabs.harvard.edu/abs/2015arXiv151107548H}
}

@ARTICLE{Renner2023,
       author = {{Renner}, Martin J. and {Tavakoli}, Armin and {Quintino}, Marco T{\'u}lio},
        title = "{Classical Cost of Transmitting a Qubit}",
      journal = {\prl},
         year = 2023,
        month = mar,
       volume = {130},
       number = {12},
          eid = {120801},
        pages = {120801},
          doi = {10.1103/PhysRevLett.130.120801},
archivePrefix = {arXiv},
       eprint = {2207.02244},
 primaryClass = {quant-ph},
       adsurl = {https://ui.adsabs.harvard.edu/abs/2023PhRvL.130l0801R}
}

@ARTICLE{Uola2019,
       author = {{Uola}, Roope and {Kraft}, Tristan and {Shang}, Jiangwei and {Yu}, Xiao-Dong and {G{\"u}hne}, Otfried},
        title = "{Quantifying Quantum Resources with Conic Programming}",
      journal = {\prl},
         year = 2019,
        month = apr,
       volume = {122},
       number = {13},
          eid = {130404},
        pages = {130404},
          doi = {10.1103/PhysRevLett.122.130404},
archivePrefix = {arXiv},
       eprint = {1812.09216},
 primaryClass = {quant-ph},
       adsurl = {https://ui.adsabs.harvard.edu/abs/2019PhRvL.122m0404U}
}

@ARTICLE{Skrzypczyk2019,
       author = {{Skrzypczyk}, Paul and {{\v{S}}upi{\'c}}, Ivan and {Cavalcanti}, Daniel},
        title = "{All Sets of Incompatible Measurements give an Advantage in Quantum State Discrimination}",
      journal = {\prl},
         year = 2019,
        month = apr,
       volume = {122},
       number = {13},
          eid = {130403},
        pages = {130403},
          doi = {10.1103/PhysRevLett.122.130403},
archivePrefix = {arXiv},
       eprint = {1901.00816},
 primaryClass = {quant-ph},
       adsurl = {https://ui.adsabs.harvard.edu/abs/2019PhRvL.122m0403S}
}

@ARTICLE{Bowles2014,
       author = {{Bowles}, Joseph and {V{\'e}rtesi}, Tam{\'a}s and {Quintino}, Marco T{\'u}lio and {Brunner}, Nicolas},
        title = "{One-way Einstein-Podolsky-Rosen Steering}",
      journal = {\prl},
         year = 2014,
        month = may,
       volume = {112},
       number = {20},
          eid = {200402},
        pages = {200402},
          doi = {10.1103/PhysRevLett.112.200402},
archivePrefix = {arXiv},
       eprint = {1402.3607},
 primaryClass = {quant-ph},
       adsurl = {https://ui.adsabs.harvard.edu/abs/2014PhRvL.112t0402B}
}

@ARTICLE{Sekatski2023,
       author = {{Sekatski}, Pavel and {Giraud}, Florian and {Uola}, Roope and {Brunner}, Nicolas},
        title = "{Unlimited One-Way Steering}",
      journal = {\prl},
         year = 2023,
        month = sep,
       volume = {131},
       number = {11},
          eid = {110201},
        pages = {110201},
          doi = {10.1103/PhysRevLett.131.110201},
archivePrefix = {arXiv},
       eprint = {2304.03888},
 primaryClass = {quant-ph},
       adsurl = {https://ui.adsabs.harvard.edu/abs/2023PhRvL.131k0201S}
}

@ARTICLE{Zhang2023,
       author = {{Zhang}, Yujie and {Chitambar}, Eric},
        title = "{Exact Steering Bound for Two-Qubit Werner States}",
      journal = {\prl},
         year = 2024,
        month = jun,
       volume = {132},
       number = {25},
          eid = {250201},
        pages = {250201},
          doi = {10.1103/PhysRevLett.132.250201},
archivePrefix = {arXiv},
       eprint = {2309.09960},
 primaryClass = {quant-ph},
       adsurl = {https://ui.adsabs.harvard.edu/abs/2024PhRvL.132y0201Z}
}

@ARTICLE{Cavalcanti2017,
       author = {{Cavalcanti}, D. and {Skrzypczyk}, P.},
        title = "{Quantum steering: a review with focus on semidefinite programming}",
      journal = {Reports on Progress in Physics},
         year = 2017,
        month = feb,
       volume = {80},
       number = {2},
          eid = {024001},
        pages = {024001},
          doi = {10.1088/1361-6633/80/2/024001},
archivePrefix = {arXiv},
       eprint = {1604.00501},
 primaryClass = {quant-ph},
       adsurl = {https://ui.adsabs.harvard.edu/abs/2017RPPh...80b4001C}
}

\onecolumngrid

\appendix

\section{A helpful lemma} \label{helplemma}
\begin{lemma}\label{lhslemma1}
    Given the eight vectors $\vec{v}_{s_x s_y s_z}$ forming a cube of sidelength two centered at the origin of the Bloch sphere and an arbitrary vector $\vec{a}\in \mathbb{R}^3$. In addition, the function $\Theta(x)$ (for $x\in \mathbb{R}$) is defined as $\Theta(x):=x$ if $x\geq0$ and $\Theta(x):=0$ if $x<0$ (equivalent: $\Theta(x):=(|x|+x)/2$). We prove the following properties:
    \begin{align}
        (1)\, &\sum_{s_x,s_y,s_z=\pm 1} |\vec{v}_{s_x s_y s_z}\cdot \vec{a}|\leq  8\cdot |\vec{a} |\\
        (2)\, &\sum_{s_x,s_y,s_z=\pm 1} \Theta(\vec{v}_{s_xs_ys_z}\cdot \vec{a})\leq  4\cdot |\vec{a}|\\
        (3)\, &\sum_{s_x,s_y,s_z=\pm 1}(\vec{v}_{s_xs_ys_z}\cdot \vec{a})\  \vec{v}_{s_xs_ys_z}=8\cdot \vec{a}\\
        (4)\, &\sum_{s_x,s_y,s_z=\pm 1}\Theta(\vec{v}_{s_xs_ys_z}\cdot \vec{a})\  \vec{v}_{s_xs_ys_z}=4\cdot \vec{a}
    \end{align}
\end{lemma}
\begin{proof}
Note that all four statements remain the same under a rotation of the coordinate system. Hence, it is sufficient to prove them in the rotated frame $\vec{x}'$, $\vec{y}'$, $\vec{z}'$ in which $\vec{v}_{s_x s_y s_z}$ has the coordinates $\vec{v}'_{s_x s_y s_z}=(s_x, s_y, s_z)^T$ where $s_x,s_y,s_z\in \{+1,-1\}$. For ease of notation, we write all vectors in the proof without the prime.\\

(1) We apply the Cauchy-Schwarz inequality to the following two eight-dimensional vectors:
\begin{align}
    \begin{pmatrix} 1 \\ 1 \\ \vdots \\ 1 \end{pmatrix} \cdot \begin{pmatrix} |\vec{v}_{+++}\cdot \vec{a}| \\ |\vec{v}_{++-}\cdot \vec{a}| \\ \vdots  \\ |\vec{v}_{---}\cdot \vec{a}| \end{pmatrix} & \leq \abs{\begin{pmatrix} 1 \\ 1 \\ \vdots \\ 1 \end{pmatrix}} \cdot \abs{\begin{pmatrix} |\vec{v}_{+++}\cdot \vec{a}| \\ |\vec{v}_{++-}\cdot \vec{a}| \\ \vdots  \\ |\vec{v}_{---}\cdot \vec{a}| \end{pmatrix}}\\
    \sum_{s_x,s_y,s_z=\pm 1} |\vec{v}_{s_x s_y s_z}\cdot \vec{a}| & \leq \sqrt{8}\cdot \sqrt{\sum_{s_x,s_y,s_z=\pm 1} |\vec{v}_{s_x s_y s_z}\cdot \vec{a}|^2}
\end{align}
Now we rewrite the right-hand side. Here we denote $\vec{a}=(a_x,a_y,a_z)^T$:
\begin{align}
    \sqrt{\sum_{s_x,s_y,s_z=\pm 1} |\vec{v}_{s_x s_y s_z}\cdot \vec{a}|^2}
    =& \sqrt{\sum_{s_x,s_y,s_z=\pm 1} (\vec{v}_{s_x s_y s_z}\cdot \vec{a})^2}\\
    =& \sqrt{(a_x+a_y+a_z)^2+(a_x+a_y-a_z)^2+...+(-a_x-a_y-a_z)^2}\\
    =& \sqrt{8\ a_x^2+8\ a_y^2+8\ a_z^2}\\
    =& \sqrt{8}\cdot \sqrt{a_x^2+a_y^2+a_z^2}\\
    =& \sqrt{8}\cdot |\vec{a} |
\end{align}
Note, that in the third line, all terms of the form $2a_xa_y$, $2a_xa_z$ or $2a_ya_z$ cancel each other out since each of these terms appear four times with a plus sign and four times with a minus sign. In total, we obtain the desired inequality:
\begin{align}
    \sum_{s_x,s_y,s_z=\pm 1} |\vec{v}_{s_x s_y s_z}\cdot \vec{a}|\leq  8\cdot |\vec{a} | \, . \label{usefulidentitiy}
\end{align}

(2) The second inequality is a consequence of the first one. 
\begin{align}
    \sum_{s_x,s_y,s_z=\pm 1} \Theta(\vec{a}\cdot \vec{v}_{s_xs_ys_z})= \frac{1}{2} \sum_{s_x,s_y,s_z=\pm 1} |\vec{a}\cdot \vec{v}_{s_xs_ys_z}|\leq \frac{8}{2} |\vec{a}|=4 |\vec{a}| \, .
\end{align}
Here, we used that:
\begin{align}
    \Theta(\vec{a}\cdot \vec{v}_{s_xs_ys_z})+\Theta(\vec{a}\cdot \vec{v}_{-s_x-s_y-s_z})=|\vec{a}\cdot \vec{v}_{s_xs_ys_z}|=\frac{1}{2}(|\vec{a}\cdot \vec{v}_{s_xs_ys_z}|+|\vec{a}\cdot \vec{v}_{-s_x-s_y-s_z}|) \, ,
\end{align}
which follows from $\Theta(x)+\Theta(-x)=|x|$ and $\vec{a}\cdot \vec{v}_{s_xs_ys_z}=-\vec{a}\cdot \vec{v}_{-s_x-s_y-s_z}$ (as a consequence of $\vec{v}_{-s_x-s_y-s_z}=-\vec{v}_{s_xs_ys_z}$).

(3) This property is a rather straightforward calculation. If we denote $\vec{a}=(a_x,a_y,a_z)^T$ and remember that 
$\vec{v}_{s_xs_ys_z}=(s_x,s_y,s_z)^T$, we obtain:
\begin{align}
    \sum_{s_x,s_y,s_z=\pm 1}(\vec{v}_{s_xs_ys_z}\cdot \vec{a})\  \vec{v}_{s_xs_ys_z}
    =&\left((a_x+a_y+a_z) \begin{pmatrix} 1 \\ 1 \\ 1 \end{pmatrix}+(a_x+a_y-a_z) \begin{pmatrix} 1 \\ 1 \\ -1 \end{pmatrix}+...+(-a_x-a_y-a_z) \begin{pmatrix} -1 \\ -1 \\ -1 \end{pmatrix}\right)\\
    =&\begin{pmatrix} 8a_x \\ 8a_y \\ 8a_z \end{pmatrix} = 8\cdot \vec{a}\, .
\end{align}

(4) For this, we note that $\vec{v}_{-s_x-s_y-s_z}=(-s_x,-s_y,-s_z)^T=-(s_x,s_y,s_z)^T=-\vec{v}_{s_xs_ys_z}$ and therefore $\vec{v}_{-s_x-s_y-s_z}\cdot \vec{a}=-\vec{v}_{s_xs_ys_z}\cdot \vec{a}$. Combining both, we obtain:
\begin{align}
    (\vec{v}_{-s_x-s_y-s_z}\cdot \vec{a})\ \vec{v}_{-s_x-s_y-s_z}=(-\vec{v}_{s_xs_ys_z}\cdot \vec{a})\ (-\vec{v}_{s_xs_ys_z})=(\vec{v}_{s_xs_ys_z}\cdot \vec{a})\ \vec{v}_{s_xs_ys_z}\, .
\end{align}
In addition, since $\Theta(x)-\Theta(-x)=x$ and again $\vec{v}_{-s_x-s_y-s_z}=-\vec{v}_{s_xs_ys_z}$ we observe:
\begin{align}
    \Theta(\vec{a}\cdot \vec{v}_{s_xs_ys_z})\ \vec{v}_{s_xs_ys_z}+\Theta(\vec{a}\cdot \vec{v}_{-s_x-s_y-s_z})\ \vec{v}_{-s_x-s_y-s_z}=&\Theta(\vec{a}\cdot \vec{v}_{s_xs_ys_z})\ \vec{v}_{s_xs_ys_z}-\Theta(-\vec{a}\cdot \vec{v}_{s_xs_ys_z})\ \vec{v}_{s_xs_ys_z}\\
    =&(\Theta(\vec{a}\cdot \vec{v}_{s_xs_ys_z})-\Theta(-\vec{a}\cdot \vec{v}_{s_xs_ys_z}))\ \vec{v}_{s_xs_ys_z}\\
    =&(\vec{a}\cdot \vec{v}_{s_xs_ys_z})\ \vec{v}_{s_xs_ys_z}\\
    =&\frac{1}{2}((\vec{a}\cdot \vec{v}_{s_xs_ys_z})\ \vec{v}_{s_xs_ys_z}+(\vec{a}\cdot \vec{v}_{-s_x-s_y-s_z})\ \vec{v}_{-s_x-s_y-s_z})
\end{align}
This implies together with property (3):
\begin{align}
    \sum_{s_x,s_y,s_z=\pm 1} \Theta(\vec{v}_{s_xs_ys_z}\cdot \vec{a})\  \vec{v}_{s_xs_ys_z} = \frac{1}{2} \sum_{s_x,s_y,s_z=\pm 1} (\vec{v}_{s_xs_ys_z}\cdot \vec{a})\  \vec{v}_{s_xs_ys_z} = 4\cdot \vec{a} \, .
\end{align}
\end{proof}

\section{Properties of the function $f_a$}\label{propertiesfunction}

The theorem by Hausel, Makai, and Szűcs \cite{hausel_makai_szűcs_2000} states that for every real-valued and continuous function on the two-sphere $S_2$ that has the additional property that it is even, i.e. $f(\vec{x})=f(-\vec{x})$, there exists an inscribed cube with all vertices lying on the sphere $S_2$ such that the function $f(\vec{x})$ has the same value on each vertex of that cube. We show now that these properties are fulfilled by the function $f_a(\vec{x})$. Strictly speaking, we apply the theorem not to the unit sphere but to the sphere with radius $|\vec{v}_{s_xs_ys_z}|=\sqrt{3}$. Alternatively, we can also apply the theorem to the unit sphere and the vectors $\vec{v}_{s_xs_ys_z}/\sqrt{3}$. Since the function satisfies $f_a(\gamma\cdot \vec{x})=\gamma \cdot  f_a(\vec{x})$ for $\gamma \geq 0$ (see below), the values of $f_a(\vec{v}_{s_xs_ys_z})$ coincide if and only if the values of $f_a(\vec{v}_{s_xs_ys_z}/\sqrt{3})$ coincide.


\begin{lemma}\label{lemmafunction}
    We consider the function $f_a: \mathbb{R}^3 \mapsto \mathbb{R}:\ f_a(\vec{x})=\sum_i p_i \ \Theta (\vec{x}\cdot \vec{a}_i)$. Here,  $|\vai |=1 \ \forall i$, $\sum_i p_i=2$ and $\sum_i p_i\  \vai=\vec{0}$. In addition, the function $\Theta(x)$ is defined as $\Theta(x):=x$ if $x\geq0$ and $\Theta(x):=0$ if $x<0$ (equivalently: $\Theta(x):=(|x|+x)/2$). The function $f_a$ satisfies:
    \begin{align}
        (1)\, f_a(\vec{x})=\frac{1}{2} \sum_{i} p_i \ |\vec{a}_i \cdot \vec{x}|&&\text{and}&&
        (2)\, \sum_{s_x, s_y, s_z =\pm 1} f_a(\vec{v}_{s_x s_y s_z})\leq 8 \, .
    \end{align}
    Here, again the eight vectors $\vec{v}_{s_x s_y s_z}$ with $s_x,s_y,s_z\in \{+1,-1\}$ form the vertices of a cube with sidelength two centered at the origin of the Bloch sphere. From property (1) we can conclude furthermore, that $f_a(\vec{x})$ is a sum of continuous functions and therefore continuous (also when restricted to the sphere $S_2$) that satisfies  $f_a(-\vec{x})=f_a(\vec{x})$ or more generally $f_a(\gamma \cdot \vec{x})=|\gamma|\cdot f_a(\vec{x})$ for every $\gamma\in \mathbb{R}$.
\end{lemma}
\begin{proof}
(1) We show that the function $f_a(\vec{x})$ can be rewritten as follows:
\begin{align}
    f_a(\vec{x})=\sum_{i} p_i \ \Theta(\vec{a}_i \cdot \vec{x})=\frac{1}{2} \sum_{i} p_i \ |\vec{a}_i \cdot \vec{x}| \, .
\end{align}
We want to remark that exactly the same property appears also in Ref.~\cite{Renner2023}  (Appendix A, Lemma 2)  and we restate the proof here: We prove first that $\sum_{i} p_i\ \Theta(\vec{a}_i\cdot \vec{x})=\sum_{i} p_i\ \Theta(-\vec{a}_i\cdot \vec{x})$. Here, we use that $x=\Theta(x)-\Theta(-x)$ (for all $x\in \mathbb{R}$):
\begin{align}
    \vec{0}=\sum_{i} p_i\ \vec{a}_i\ \implies \
    0=\vec{0}\cdot \vec{x}=\sum_{i} p_i\ \vec{a}_i\cdot \vec{x}=\sum_{i} p_i\ (\Theta(\vec{a}_i\cdot \vec{x})-\Theta(-\vec{a}_i\cdot \vec{x}))=\sum_{i} p_i\ \Theta(\vec{a}_i\cdot \vec{x})-\sum_{i} p_i\ \Theta(-\vec{a}_i\cdot \vec{x}) \, .
\end{align}
In the second step, we use this observation and $|x|=\Theta(x)+\Theta(-x)$ (for all $x\in \mathbb{R}$) to calculate:
\begin{align}
    \sum_i p_i \ |\vec{a}_i \cdot \vec{x}|=\sum_i p_i \ (\Theta(\vec{a}_i \cdot \vec{x})+\Theta(-\vec{a}_i \cdot \vec{x}))=\sum_i p_i \ \Theta(\vec{a}_i \cdot \vec{x})+\sum_i p_i \ \Theta(-\vec{a}_i \cdot \vec{x})=2 \sum_i p_i \ \Theta(\vec{a}_i \cdot \vec{x}) \, .
\end{align}

(2) This is a direct consequence of property (1) in Lemma~\ref{lhslemma1} we proved above together with multiplying by $p_i$ and taking the sum over all $i$:
\begin{align}
    2\ \sum_{s_x, s_y, s_z =\pm 1} f_a(\vec{v}_{s_x s_y s_z})&= \sum_{s_x, s_y, s_z =\pm 1} \sum_i p_i\ |\vec{v}_{s_x s_y s_z}\cdot \vec{a}_i|= \sum_i p_i \sum_{s_x, s_y, s_z =\pm 1} |\vec{v}_{s_x s_y s_z}\cdot \vec{a}_i|\leq \sum_i p_i\cdot  8 \cdot |\vai |=16 \, .
\end{align}
Note that $|\vai|=1$ for all $i$ and $\sum_i p_i=2$.

\end{proof}

\section{Proof of the protocol}\label{identity}
In this section, we show that the protocol indeed simulates the noisy POVM with elements $A^{1/2}_{i|a}$. We have to show that
\begin{align}
    \int_{S_2} \mathrm{d}\lam \ p(i|a, \lam) G_{\lam} =\frac{p_i}{2}\left(\id +\frac{\vai\cdot \vec{\sigma}}{2}\right)=A^{1/2}_{i|a} \, .
\end{align}
We also restate the definitions from the main text here:
\begin{align}
    p(i|a, \lam)&:=p_i \cdot \Theta(\vec{a}_i\cdot \vec{v}_{s_xs_ys_z})+\frac{(1-f_a(\vec{v}_{s_xs_ys_z}))\alpha_i}{\sum_i \alpha_i}\, ,\\
    \alpha_i&:=\frac{p_i}{2}\left(1-\frac{1}{4}\sum_{s_x,s_y,s_z=\pm 1} \Theta(\vec{a}_i\cdot \vec{v}_{s_xs_ys_z})\right) \, .
\end{align}
It is important to recognize that the function $p(i|a, \lam)$ is constant in each octant of the rotated coordinate frame $x', y', z'$ since $\vec{v}_{s_xs_ys_z}$ is given by $s_k=\sgn{(\lam \cdot \vec{k}')}$ for $k\in \{x,y,z\}$. Intuitively speaking, we collect all the measurement results from one octant of the rotated frame together. This coarse-graining of the parent POVM can be calculated when we integrate over all vectors $\lam$ in the corresponding octant. We denote this as the operator $G_{s_xs_ys_z}$ that becomes (calculation in the next Subsection~\ref{integrals}):
\begin{align}
    G_{s_xs_ys_z}:=\int_{S_2|\sgn{(\lam \cdot \vec{k}')}=s_k} \mathrm{d}\lam \ G_{\lam} =\frac{1}{16}(2\cdot \id+\vec{v}_{s_xs_ys_z}\cdot \vec{\sigma})\, .
\end{align}
This operator behaves like a noisy measurement in the direction of the corresponding vector $\vec{v}_{s_xs_ys_z}$. Using this, the above integration reduces to:
\begin{align}
    \int_{S_2} \mathrm{d}\lam \ p(i|a, \lam)\ G_{\lam} =\sum_{s_x, s_y, s_z =\pm 1}  p(i|a, \lam)\  G_{s_xs_ys_z} \, ,
\end{align}
where we again note that the conditional probabilities $p(i|a, \lam)$ only depend on the signs of $\lam$ in the rotated frame. Now we can evaluate the right-hand side of this equation. We obtain:
\begin{align}
    \sum_{s_x, s_y, s_z =\pm 1}&  p(i|a, \lam)\  G_{s_xs_ys_z}=\sum_{s_x, s_y, s_z =\pm 1}p_i \ \Theta(\vec{a}_i\cdot \vec{v}_{s_xs_ys_z})\ G_{s_xs_ys_z} + \sum_{s_x, s_y, s_z =\pm 1} \frac{(1-f_a(\vec{v}_{s_xs_ys_z}))\alpha_i}{\sum_i \alpha_i}\ G_{s_xs_ys_z} \, .\label{twoterms}
\end{align}
We evaluate the first term first:
\begin{align}
    \sum_{s_x,s_y,s_z=\pm 1}  p_i  \ \Theta(\vec{a}_i\cdot \vec{v}_{s_xs_ys_z})\ G_{s_xs_ys_z}&=\frac{1}{8}\sum_{s_x,s_y,s_z=\pm 1} p_i \ \Theta(\vec{a}_i\cdot \vec{v}_{s_xs_ys_z})\ \id +\frac{1}{16}\sum_{s_x,s_y,s_z=\pm 1} p_i \ \Theta(\vec{a}_i\cdot \vec{v}_{s_xs_ys_z}) \vec{v}_{s_xs_ys_z}\cdot \vec{\sigma}\\
    &=\frac{1}{8}\sum_{s_x,s_y,s_z=\pm 1} p_i \ \Theta(\vec{a}_i\cdot \vec{v}_{s_xs_ys_z})\ \id +\frac{1}{4}p_i \ \vec{a}\cdot \vec{\sigma}\\
    &=\frac{p_i}{2}\left(\id +\frac{\vai\cdot \vec{\sigma}}{2}\right)-\alpha_i \id =A^{1/2}_{i|a}-\alpha_i \id \, .\label{ide}
\end{align}
Here we used property (4) in Lemma~\ref{lhslemma1} and the definition of $\alpha_i$ to rewrite the expression into the desired form. The second term becomes:
\begin{align}
    \sum_{s_x, s_y, s_z =\pm 1} \frac{(1-f_a(\vec{v}_{s_xs_ys_z}))\alpha_i}{\sum_i \alpha_i}\ G_{s_xs_ys_z}=\frac{\sum_{s_x, s_y, s_z =\pm 1}(1-f_a(\vec{v}_{s_xs_ys_z}))\alpha_i}{8\sum_i \alpha_i}\ \id=\alpha_i \id \, . \label{secterm}
\end{align}
To see this, we note that the coefficient in front of $G_{s_xs_ys_z}$ and $G_{-s_x-s_y-s_z}$ are the same since $f_a(\vec{v}_{s_xs_ys_z})=f_a(-\vec{v}_{s_xs_ys_z})=f_a(\vec{v}_{-s_x-s_y-s_z})$ and the $\alpha_i$ do not depend on $s_k$. In addition, $G_{s_xs_ys_z}+G_{-s_x-s_y-s_z}=\id/4=\id/8+\id/8$ which allows us to replace each $G_{s_xs_ys_z}$ by $\id/8$ in the first step. Furthermore, we observe that:
\begin{align}
    8\sum_i \alpha_i&=8\sum_i \frac{p_i}{2}\left(1-\frac{1}{4}\sum_{s_x,s_y,s_z=\pm 1} \Theta(\vec{a}_i\cdot \vec{v}_{s_xs_ys_z})\right)\\
    &=8-\sum_i\sum_{s_x,s_y,s_z=\pm 1} p_i\ \Theta(\vec{a}_i\cdot \vec{v}_{s_xs_ys_z})\\
    &=\sum_{s_x, s_y, s_z =\pm 1} (1- f_a(\vec{v}_{s_xs_ys_z})) \, ,
\end{align}
which explains the last step in Eq.~\eqref{secterm}. Hence, the sum of the two terms in Eq.~\eqref{twoterms} is exactly $A^{1/2}_{i|a}$ as required.

\subsection{Dividing the sphere into the eight octants}\label{integrals}
We divided the sphere into eight regions, according to the eight octants defined by the signs in the rotated coordinate system. Our goal here is to show that when we integrate all $\lam$ in one octant, the result is precisely  $G_{s_xs_ys_z}=(2\cdot \id+\vec{v}_{s_xs_ys_z}\cdot \vec{\sigma})/16$, where $\vec{v}_{s_xs_ys_z}$ is the vertex of the cube that points to the midpoint of the octant. We do the calculation first in the standard coordinate frame (without rotation) for the vector $\vec{v}_{+++}=(+1,+1,+1)^T$ and show that $G_{+++}=(2\cdot \id+\vec{v}_{+++}\cdot \vec{\sigma})/16$. Here we use spherical coordinates $\lam=(\sin{\theta}\cos{\phi},\sin{\theta}\sin{\phi},\cos{\theta})$ and integrate the parent POVM $G_{\vec{\lambda}}=\frac{1}{4\pi}(\id + \vec{\lambda}\cdot \vec{\sigma})$ over all $\lam$ in that octant. Hence, all $\lam$ that have only positive components. This is true if and only if $0\leq \phi \leq \pi/2$ as well as $0\leq \theta \leq \pi/2$:
\begin{align}
    G_{+++}&= \frac{1}{4\pi}\int_0^{\pi/2} \mathrm{d}\theta \int_{0}^{\pi/2} \mathrm{d}\phi\ (\id+\sigma_x \sin{\theta}\cos{\phi} + \sigma_y \sin{\theta}\sin{\phi} +\sigma_z \cos{\theta}) \sin{\theta}\\
    &=\frac{1}{8}\id+\frac{1}{16} \left(\sigma_x+\sigma_y+\sigma_z\right) \, .
\end{align}
However, in general, the integration is in a rotated frame. Nevertheless, the shape of an octant is always the same, only the position is rotated. Hence, by symmetry arguments, we can conclude that $G_{s_xs_ys_z}=(2\cdot \id+\vec{v}_{s_xs_ys_z}\cdot \vec{\sigma})/16$ holds for a general octant corresponding to the cube vertex $\vec{v}_{s_xs_ys_z}$.

\section{Special cases}\label{specialcases}
Here, we discuss some special cases and provide more illustrations. The case of two-outcome POVMs is already discussed in Appendix~\ref{appendixprl}.

\subsection{Three-outcome measurements}
Suppose that all the vectors $\vai$ lie in a plane. In particular, this is true if the POVM has only three outcomes since $p_1\vec{a}_1+p_2\vec{a}_2+p_3\vec{a}_3=\vec{0}$ can only be satisfied if all three vectors lie in the same plane. In that case, we can choose the basis such that $z'$ is orthogonal to the plane in which the vectors lie. With that choice of coordinate frame, we can observe that $f_a(\vec{v}_{s_xs_y+})=f_a(\vec{v}_{s_xs_y-})$ since $\vec{a}_i \cdot \vec{z}'=0$ and therefore the $z'$-component of $\vec{v}_{s_xs_ys_z}$ does not affect the value of $f_a(\vec{v}_{s_xs_ys_z})$. Together with $f_a(\vec{v}_{s_xs_ys_z})=f_a(\vec{v}_{-s_x-s_y-s_z})$ (note that $\vec{v}_{s_xs_ys_z}=-\vec{v}_{-s_x-s_y-s_z}$ and $f_a(\vec{x})=f_a(-\vec{x})$), we can denote $C_1:=f_a(\vec{v}_{+++})=f_a(\vec{v}_{++-})=f_a(\vec{v}_{--+})=f_a(\vec{v}_{---})$ and $C_2:=f_a(\vec{v}_{+-+})=f_a(\vec{v}_{+--})=f_a(\vec{v}_{-++})=f_a(\vec{v}_{-+-})$. As a consequence of property (2) in Lemma~\ref{lemmafunction}, we obtain $C_1+C_2\leq 2$. Now we can show that there always exists a rotation around the $z'$-axis such that both values $C_1$ and $C_2$ are smaller or equal to one. If we fix at the beginning a coordinate frame where both values are smaller than one, we can use precisely that frame. On the other hand, if one value (suppose $C_1$) is above one, the other one ($C_2$) is smaller than one. Now we rotate the coordinate axes $x'$ and $y'$ around the $z'$-axis. If we rotate by 90 degrees, we map the vector $\vec{v}_{+-+}$ to the vector $\vec{v}_{+++}$ and therefore in the rotated coordinate frame we obtain $C^{*}_1=f_a(\vec{v}^{*}_{+++})=f_a(\vec{v}_{+-+})=C_2\leq 1$. By the intermediate value theorem and since $f_a$ is continuous, there is a rotation (with less than 90 degrees) such that $C_1=f_a(\vec{v}^{*}_{+++})=1$ which implies that $C^{*}_2=f_a(\vec{v}^{*}_{+-+})\leq 1$ since $C^{*}_1+C^{*}_2\leq 2$ holds for each of these coordinate frames. In this way, we can construct a suitable coordinate system without relying on the theorem of Ref.~\cite{hausel_makai_szűcs_2000} but only on the intermediate value theorem.

\begin{figure}[hbt!]
    \centering
    \includegraphics[width=1.0\columnwidth]{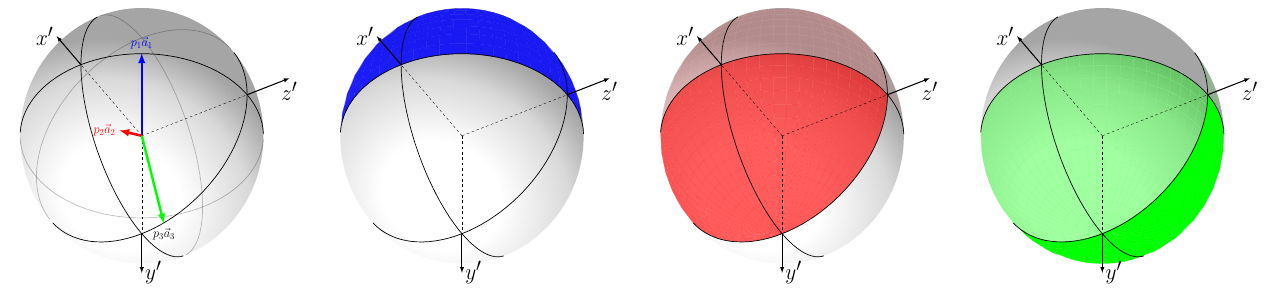}
    \caption{An illustration of $p(i|a,\lam)$ for a three-outcome POVM. Here the conditional probabilities do not depend on $z'$ due to the choice of the coordinate frame. If $\lam$ is close to one of the colored vectors it is also more likely that this color is produced as an output.}
\end{figure}


\subsection{SIC-POVM}\label{SICexample}
We also want to give an example with a four-outcome measurement, namely a SIC-POVM. We can represent a SIC POVM as follows:
\begin{align}
    \vec{a}_1&=\begin{pmatrix} 0 \\ 0 \\ 1 \end{pmatrix}&
    \vec{a}_2&=\begin{pmatrix} \sqrt{8}/3 \\ 0 \\ -1/3 \end{pmatrix}&
    \vec{a}_3&=\begin{pmatrix} -\sqrt{2}/3 \\ \sqrt{6}/3 \\ -1/3 \end{pmatrix}&
    \vec{a}_4&=\begin{pmatrix} -\sqrt{2}/3 \\ -\sqrt{6}/3 \\ -1/3 \end{pmatrix}
\end{align}
and the coefficients $p_i$ are $p_1=p_2=p_3=p_4=1/2$. It turns out, that we can use the standard coordinate frame and no rotation of the basis is necessary. In that basis, the eight vertices of the cube become $\vec{v}_{s_xs_ys_z}=(s_x,s_y,s_z)^T$ with $s_x,s_y,s_z \in \{+1,-1\}$ and we can indeed verify that $f_a(\vec{v}_{s_xs_ys_z})\leq 1$ for each $\vec{v}_{s_xs_ys_z}$. See the following table:
\begin{table}[hbt!]
    \centering
    \begin{tabular}{c|ccc||c|c|c|c|c|c|c|c||c|c}
         i & &$\vec{a}_i$& & + + + & + + - & + - + & + - - & - + + & - + - & - - + & - - - & $\sum$ & $\alpha_i$\\\hline\hline
        1 (blue) & (0.000, & 0.000, & $1.000)^T$ & 0.5 & 0 & 0.5 & 0 & 0.5 & 0 & 0.5 & 0 & 2 & 0  \\ \hline
        2 (red) & (0.943, & 0.000, & $-0.333)^T$ & 0.305 & 0.638 & 0.305 & 0.638 & 0 & 0 & 0 & 0 & 1.886 & 0.014  \\ \hline
        3 (green) & (-0.471, & 0.816, & $-0.333)^T$ & 0.006 & 0.339 & 0 & 0 & 0.477 & 0.811 & 0 & 0 & 1.633 & 0.046  \\ \hline
        4 (yelow) & (-0.471, & -0.816, & $-0.333)^T$ & 0 & 0 & 0.006 & 0.339 & 0 & 0 & 0.477 & 0.811 & 1.633 & 0.046  \\ \hline\hline
        & $f_a(\vec{v}_{s_xs_ys_z})$&$=\sum_i p_i$ &$ \Theta(\vec{v}_{s_xs_ys_z}\cdot \vec{a}_i)$& 0.811 & 0.977 & 0.811 & 0.977 & 0.977 & 0.811 & 0.977 & 0.811 & 7.152 &
    \end{tabular}
    \caption{Here we represent the functional values of $p_i\ \Theta(\vec{v}_{s_xs_ys_z}\cdot \vec{a}_i)$ for the chosen SIC-POVM (note that $p_i=1/2$ for every $i$). The last row gives the values for the function $f_a(\vec{v}_{s_xs_ys_z})=\sum_i p_i\ \Theta(\vec{v}_{s_xs_ys_z}\cdot \vec{a}_i)$ which is obtained by taking the sum over all $i$. We also calculate every $\alpha_i$ in the last column.}
    \label{tab:my_label}
\end{table}

\begin{table}[hbt!]
    \centering
    \begin{tabular}{c||c|c|c|c|c|c|c|c||c}
         i & + + + & + + - & + - + & + - - & - + + & - + - & - - + & - - - & $\sum$ \\\hline\hline
        1 (blue) & 0.5 & 0 & 0.5 & 0 & 0.5 & 0 & 0.5 & 0  & 2 \\ \hline
        2 (red) & 0.330 & 0.641 & 0.330 & 0.641 & 0.003 & 0.026 & 0.003 & 0.026  & 2 \\ \hline
        3 (green) & 0.088 & 0.349 & 0.082 & 0.010 & 0.487 & 0.893 & 0.010 & 0.082  & 2 \\ \hline
        4 (yellow) & 0.082 & 0.010 & 0.088 & 0.349 & 0.010 & 0.082 & 0.487 & 0.893  & 2 \\ \hline\hline
        & 1 & 1 & 1 & 1 & 1 & 1 & 1 & 1  &
    \end{tabular}
    \caption{Here we represent the conditional probabilities $p(i|a,\lam)$ given the octant in which $\lam=(\lambda_x,\lambda_y,\lambda_z)^T$ lies (denoted as $s_x s_y s_z$ where $s_k=\sgn(\lambda_k)$ for $k=x,y,z$). Intuitively speaking, they are the same values as in the table above but we fill the rest with noise. The last row is obtained by taking the sum over all $i$ which shows that the probabilities sum to one. A short calculation can also directly verify that $\sum_{s_x,s_y,s_z=\pm 1} p(i|a,\lam) \ G_{s_xs_ys_z}=A^{1/2}_{i|a}$.}
    \label{tab:my_label2}
\end{table}

\begin{figure}[hbt!]
    \centering
    \includegraphics[width=0.65\columnwidth]{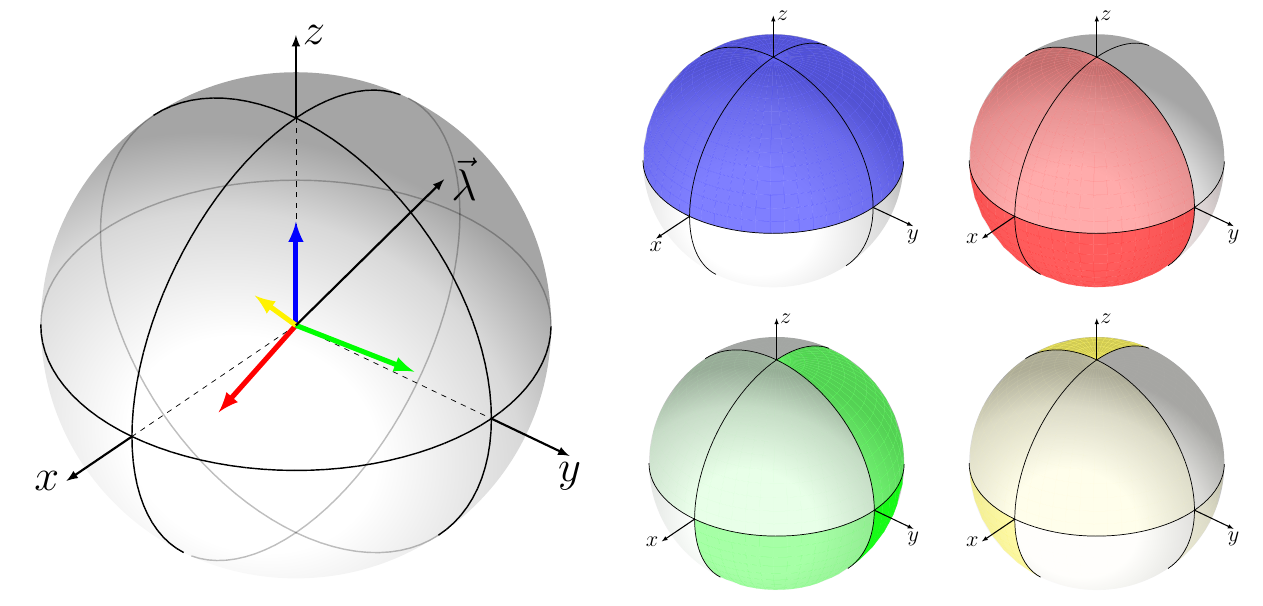}
    \caption{(similar figure as in the main text) The coloured arrows denote the vectors $p_i\cdot  \vec{a}_i$ according to $p_i=1/2$ and the vectors $\vai$ given above. The right part of the figure represents the conditional probabilities given in the table above. Here, $\lam$ lies in the octant that corresponds to "-++" ($s_x=-1$, $s_y=s_z=+1$) the outcome is $i=1$ (blue) with $p(1|a,\lam)=0.5$, $i=2$ (red) with $p(2|a,\lam)=0.003$, $i=3$ (green) with $p(3|a,\lam)=0.487$, and $i=4$ (yellow) with $p(4|a,\lam)=0.01$. Therefore, the outcome is most likely to be either $i=1$ or $i=3$.} \label{illSIC}
\end{figure}

\subsubsection{Every orthonormal basis is suitable for SIC-POVMs}\label{appsicevery}
For the construction of the SIC-POVM above, we simply used the standard coordinate frame. However, it turns out that any other choice of orthonormal basis (or any rotation of the cube) would be equally valid. Therefore, in a different coordinate frame, the functions for the conditional probabilities would change but a simulation is still possible. Indeed, we can prove that $f_a(\vec{x})\leq 1$ for every vector $\vec{x}$ with $|\vec{x}|\leq \sqrt{3}$. Therefore for every rotation of the cube, the vertices $\vec{v}_{s_xs_ys_z}$ satisfy $f_a(\vec{v}_{s_xs_ys_z})\leq 1$ since $|\vec{v}_{s_xs_ys_z}|=\sqrt{3}$. For the SIC-POVM the function becomes the following:
\begin{align}
    f_a(\vec{x})=\sum_i p_i\ \Theta(\vec{x}\cdot \vai) = \frac{\Theta(\vec{x}\cdot \vec{a}_1)+\Theta(\vec{x}\cdot \vec{a}_2)+\Theta(\vec{x}\cdot \vec{a}_3)+\Theta(\vec{x}\cdot \vec{a}_4)}{2}
\end{align}
Depending on the region where $\vec{x}$ lies, some of the terms $\vec{x}\cdot \vec{a}_i$ are positive and some of them are negative. We show that in any case, $f_a(\vec{x})\leq 1$ if $|\vec{x}|\leq \sqrt{3}$.
Suppose only the term $\vec{x}\cdot \vec{a}_1$ is positive and the remaining three terms $\vec{x}\cdot \vec{a}_i$ are negative (this happens for instance if $\vec{x}=(x,0,0)$). If this is the case, the function becomes 
\begin{align}
    f_a(\vec{x})=\frac{\Theta(\vec{x}\cdot \vec{a}_1)}{2}=\frac{\vec{x}\cdot \vec{a}_1}{2}\leq \frac{|\vec{x}|\cdot |\vec{a}_1|}{2}=\frac{|\vec{x}|}{2} \, ,
\end{align}
where we used the Cauchy-Schwarz inequality and $|\vec{a}_1|=1$. The same argument holds if one of the other terms is positive and the remaining three are negative. Now consider, that the two terms $\vec{x}\cdot \vec{a}_1$ and $\vec{x}\cdot \vec{a}_2$ are positive, and the remaining two are negative. Then the function becomes:
\begin{align}
    f_a(\vec{x})&=\frac{\Theta(\vec{x}\cdot \vec{a}_1)+\Theta(\vec{x}\cdot \vec{a}_2)}{2}=\frac{\vec{x}\cdot \vec{a}_1+\vec{x}\cdot \vec{a}_2}{2}=\frac{\vec{x}\cdot (\vec{a}_1+\vec{a}_2)}{2}\leq \frac{|\vec{x}|\cdot |\vec{a}_1+\vec{a}_2|}{2}=\frac{|\vec{x}|\cdot \sqrt{\frac{4}{3}}}{2}=\frac{|\vec{x}|}{\sqrt{3}} \, .
\end{align}
Here we used again the Cauchy-Schwarz inequality and note that $\vec{a}_1+\vec{a}_2=(\sqrt{8}/3,0,2/3)^T$ (hence $|\vec{a}_1+\vec{a}_2|=\sqrt{4/3}$). Due to symmetry reasons (or by a similar calculation), the same applies to any other combination of these four terms, in which exactly two of them are positive.

In the case where three terms are positive, we obtain similarly (note that $\vec{a}_1+\vec{a}_2+\vec{a}_3+\vec{a}_4=\vec{0}$):
\begin{align}
    f_a(\vec{x})&=\frac{\Theta(\vec{x}\cdot \vec{a}_1)+\Theta(\vec{x}\cdot \vec{a}_2)+\Theta(\vec{x}\cdot \vec{a}_3)}{2}=\frac{\vec{x}\cdot (\vec{a}_1+\vec{a}_2+\vec{a}_3)}{2}=\frac{\vec{x}\cdot (-\vec{a}_4)}{2}\leq \frac{|\vec{x}|\cdot |\vec{a}_4|}{2}=\frac{|\vec{x}|}{2} \, .
\end{align}
If none of the terms is positive, the function becomes $f_a(\vec{x})=0$. If all of the terms are positive, the function becomes also $f_a(\vec{x})=(\vec{x}\cdot (\vec{a}_1+\vec{a}_2+\vec{a}_3+\vec{a}_4))/2=0$ since $\vec{a}_1+\vec{a}_2+\vec{a}_3+\vec{a}_4=\vec{0}$. (However, there are no vectors except $\vec{x}=\vec{0}$ where either none or all of the terms are positive due to geometric arguments.) One can see that, no matter in which case we are, for every vector with $|\vec{x}|\leq \sqrt{3}$ the function satisfies $f_a(\vec{x})\leq 1$, and therefore every orthonormal frame (respectively, any rotation of the cube) can be chosen.

\end{document}